\documentclass[conference]{IEEEtran}
\IEEEoverridecommandlockouts
\usepackage{cite}
\usepackage{amsmath,amssymb,amsfonts}
\usepackage{graphicx}
\def\BibTeX{{\rm B\kern-.05em{\sc i\kern-.025em b}\kern-.08em
    T\kern-.1667em\lower.7ex\hbox{E}\kern-.125emX}}

\usepackage{subfigure}
\usepackage{epstopdf}
\usepackage{epsfig}
\usepackage{comment}
\usepackage[ruled,linesnumbered,vlined]{algorithm2e}
\SetKw{Continue}{continue}
\usepackage{amsthm}
\usepackage{accents}
\usepackage{booktabs}
\usepackage{tabularx}

\usepackage{url}

\newtheorem{theorem}{Theorem}
\newtheorem{lemma}{Lemma}
\newtheorem{definition}{Definition}

\newcommand{\ubar}[1]{\underaccent{\bar}{#1}}

\usepackage{mathtools}

\newcommand{\ie}{{\em i.e.}}
\newcommand{\eg}{{\em e.g.}}
\newcommand{\et}{{\em et al.}}

\newcommand{\resp}{{\em resp.}}

\newcommand{\oie}{i.e.}

\begin{document}

\title{Online Pricing with Reserve Price Constraint for Personal Data Markets}

\author{
\IEEEauthorblockN{Chaoyue~Niu, Zhenzhe~Zheng, Fan~Wu, Shaojie~Tang$^\dag$, and Guihai~Chen}
\IEEEauthorblockA{Shanghai Key Laboratory of Scalable Computing and Systems, Shanghai~Jiao~Tong~University, China\\
$^\dag$Department of Information Systems, University of Texas at Dallas, USA\\
Email: \{rvince, zhengzhenzhe, wu-fan\}@sjtu.edu.cn; tangshaojie@gmail.com; gchen@cs.sjtu.edu.cn
}
}

\maketitle

\begin{abstract}
The society's insatiable appetites for personal data are driving the emergency of data markets, allowing data consumers to launch customized queries over the datasets collected by a data broker from data owners. In this paper, we study how the data broker can maximize her cumulative revenue by posting reasonable prices for sequential queries. We thus propose a contextual dynamic pricing mechanism with the reserve price constraint, which features the properties of ellipsoid for efficient online optimization, and can support linear and non-linear market value models with uncertainty. In particular, under low uncertainty, our pricing mechanism provides a worst-case regret logarithmic in the number of queries. We further extend to other similar application scenarios, including hospitality service, online advertising, and loan application, and extensively evaluate three pricing instances of noisy linear query, accommodation rental, and impression over MovieLens 20M dataset, Airbnb listings in U.S. major cities, and Avazu mobile ad click dataset, respectively. The analysis and evaluation results reveal that our proposed pricing mechanism incurs low practical regret, online latency, and memory overhead, and also demonstrate that the existence of reserve price can mitigate the cold-start problem in a posted price mechanism, and thus can reduce the cumulative regret.
\end{abstract}



\section{Introduction}

With the proliferation of Internet of Things (IoTs), tremendous volumes of data are collected to monitor human behaviors in daily life. However, for the sake of security, privacy, or business competition, most of data owners are reluctant to share their data, resulting in a large number of data islands. The data isolation status locks the value of personal data against potential data consumers, such as commercial companies, financial institutions, medical practitioners, and researchers. To facilitate personal data circulation, more and more data brokers have emerged to build bridges between the data owners and the data consumers. Typical data brokers in industry include Factual~\cite{link:factual}, DataSift~\cite{link:datasift}, Datacoup~\cite{link:datacoup}, CitizenMe~\cite{link:citizenme}, and CoverUS~\cite{link:coverus}. On one hand, a data broker needs to adequately compensate the privacy leakages of data owners during the usage of their data, and thus incentivize them to contribute private data. On the other hand, the data broker should properly charge the online data consumers for their sequential queries over the collected datasets, since the behaviors of both underpricing and overpricing can incur the loss of revenue at the data broker. Such a data circulation ecosystem is conventionally called ``data market" in the literature~\cite{jour:vldb11:datamarket}.


\begin{figure}[t]
\centering
\includegraphics[width = 0.9\columnwidth]{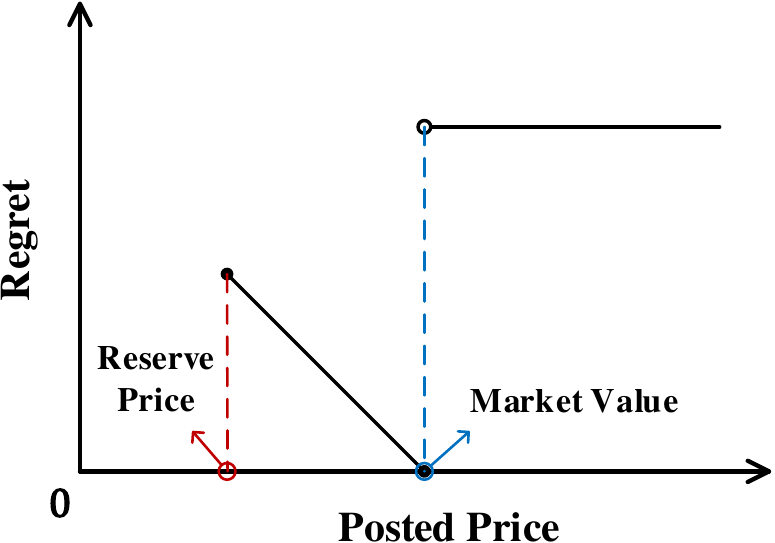}
\caption{Single-round regret function of a posted price mechanism with reserve price constraint, if the reserve price is no more than the market value.}\label{fig:regret:illustration}
\end{figure}

In this paper, we study how to trade personal data for revenue maximization from the data broker's standpoint in online data markets. We summarize three major design challenges as follows. The first and the thorniest challenge is that the objective function for optimization is quite complicated. The principal goal of a data broker in data markets is to maximize her cumulative revenue, which is defined as the difference between the prices of queries charged from the data consumers and the privacy compensations allocated to the data owners. Let's examine one round of data trading as follows. Given a query, the privacy leakages together with the total privacy compensation, regarded as the reserve price of the query, are virtually fixed. Thus, for revenue maximization, an ideal way for the data broker is to post a price, which takes the larger value of the query's reserve price and market value. However, the reality is that the data broker does not know the exact market value, and can only estimate it from the context of the current query and the historical transaction records. Of course, loose estimations will lead to different levels of regret: if the reserve price is higher than the market value, the query definitely cannot be sold, and the regret is zero; if the reserve price is no more than the market value, as shown in Fig.~\ref{fig:regret:illustration}, a slight underestimation of the market value incurs a low regret, whereas a slight overestimation causes the query not to be sold, generating a high regret. Therefore, the initial goal of revenue maximization can be equivalently converted to regret minimization. Considering even the single-round regret function is piecewise and highly asymmetric, it is nontrivial for the data broker to perform optimization for multiple rounds.


Yet, another challenge lies in how to model the market values of the customized queries from the data consumers. To minimize the regret in pricing online queries, the pivotal step for the data broker is to gain a good knowledge of their market values. However, markets for personal data significantly differ from conventional markets in that each data consumer as a buyer, rather than the data broker as a seller, can determine the product, namely a query. In general, each query involves a concrete data analysis method and a tolerable level of noise added to the true answer, which are both customized by a data consumer \cite{jour:cacm:2017:roth,jour:cacm:2017:li}. Hence, the queries from different data consumers are highly differentiated, and are uncontrollable by the data broker. This striking property further implies that most of the dynamic pricing mechanisms, which target identical products or a manageable number of distinct products, cannot apply here. Besides, existing works on data pricing, which either considered a single query \cite{proc:ec11:ghosh:selling:privacy} or investigated the determinacy relation among multiple queries \cite{proc:pods:2012:koutris,jour:vldb:2012:arbitrage,proc:sigmod:2013:arbitrage,jour:vldb:2014:lin,proc:sigmod:17:arbitrage,proc:icdt:2017:deep:arbitrage,jour:vldb:2017:arbitrage,jour:cacm:2017:li,proc:kdd:2018:erato,proc:infocom:2019:horae}, but ignored whether the data consumers accept or reject the marked prices, and thus omitted modeling the market values of queries, are parallel to this work.


The ultimate challenge comes from the novel online pricing with reserve price setting. For the market value estimation of a query, the data broker can only exploit the current and historical queries. Thus, the pricing of sequential queries can be viewed as an online learning process. In addition to the usual tension between exploitation and exploration, our pricing problem also needs to incorporate three atypical aspects. First, the feedback after trading one query is very limited. In particular, the data broker can only observe whether the posted price for the query is higher than its market value or not, but cannot obtain the exact market value, which makes standard online learning algorithms \cite{jour:ftml:2012,jour:jto:2016:oco} inapplicable. Second, the reserve price essentially imposes a lower bound on the posted price beyond the market value estimation, while the ordering between the reserve price and the market value is unknown. Besides, the impact of such a lower bound on the whole learning process has not been studied. Last but not least, the online mode requires our design of the posted price mechanism to be quite efficient. In other words, the data broker needs to choose each posted price and further update her knowledge about the market value model with low latency.

\begin{figure*}[t]
\centering
\includegraphics[width = 1.9\columnwidth]{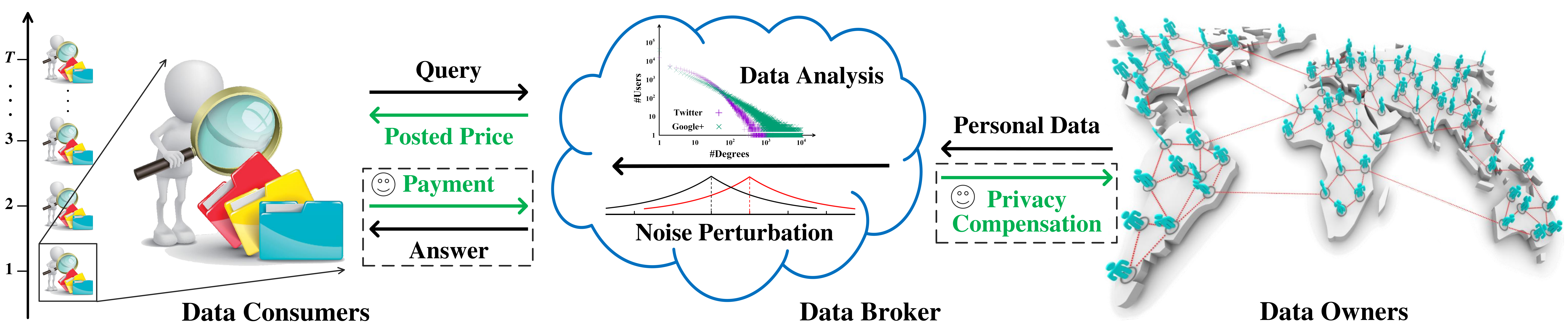}
\caption{A general system model of online personal data markets. (The smile indicates that the posted price is accepted and a deal is made.)}\label{fig:system:model}
\end{figure*}


Jointly considering the above three challenges, we propose a contextual dynamic pricing mechanism with the reserve price constraint for the data broker to maximize her revenue in online personal data markets. For problem formulation, we first adopt contextual/hedonic pricing to model the market values of different queries, which are a certain linear or non-linear function of their features plus some uncertainty. Besides, we choose the state of the privacy compensations under a query as its feature vector. In fact, such a feature representation inherits the key principle of cost-plus pricing. For posted price mechanism design, we start with the fundamental linear model, and covert the market value estimation problem to dynamically exploiting and exploring the market values of different features, \ie, the weight vector in the linear model. Specifically, depending on whether a sale occurs or not in each round, the data broker can introduce a linear inequality to update her knowledge set about the weight vector. Thus, the raw knowledge set is kept in the shape of polytope, which makes the real-time task of predicting the range of a query's market value computationally infeasible. To handle this problem, we replaces the raw knowledge set with its smallest enclosing ellipsoid, namely L\"{o}wner-John ellipsoid. Under the ellipsoid-shaped knowledge set, it only requires a few matrix-vector and vector-vector multiplications to obtain a lower bound and an upper bound on each query's market value. By further incorporating the total privacy compensation, namely the reserve price, as an additional lower bound, we define a conservative posted price and an exploratory posted price for a query. In particular, using the conservative price, the data broker sells the query with the highest probability, but cannot refine the knowledge set, and thus may not benefit the subsequent rounds. In contrast, the exploratory price, which is inspired by bisection, suffers a higher risk of no sale, but can narrow down the knowledge set by most, and thus may benefit the subsequent rounds most. In a nutshell, these two kinds of posted prices give different biases to the immediate rewards (exploitation) and the future rewards (exploration). Besides, the choice of which price in a certain round hinges on the size measure of the latest knowledge set. We further investigate how to tolerate uncertainty, and mainly introduce a ``buffer" in posting the price and updating the knowledge set. We finally extend to several non-linear models commonly used in interpreting market values, including log-linear, log-log, logistic, and kernelized models. For other similar application scenarios, we outline the characteristics of online personal data markets, including customization, existence of reserve price, and timeliness, and further illustrate the extensions to hospitality service on booking platforms, online advertising on web publishers, and loan application from banks and other financial institutions.



We list our key contributions in this paper as follows.

$\bullet$~To the best of our knowledge, we are the first to study trading personal data for revenue maximization, from the data broker's point of view in online data markets. Additionally, we formulate this problem into a contextual dynamic pricing problem with the reserve price constraint.

$\bullet$~Our proposed pricing mechanism features the properties of ellipsoid to exploit and explore the market values of sequential queries effectively and efficiently. It facilitates both linear and non-linear market value models, and is robust to some uncertainty. In particular, the worst-case regret under low uncertainty is $O(\max(n^2\log(T/n), n^3\log(T/n)/T))$, where $n$ is the dimension of feature vector and $T$ is the total number of rounds. Besides, the time and space complexities are $O(n^2)$. Furthermore, our market framework can also support trading other similar products, which share common characteristics with online queries.


$\bullet~$We extensively evaluate three application instances over three real-world datasets. The analysis and evaluation results reveal that our pricing mechanism incurs low practical regret, online latency, and memory overhead, under both linear and non-linear market value models and over both sparse and dense feature vectors. In particular, (1) for the pricing of noisy linear query under the linear model, when $n = 100$ and the number of rounds $t$ is $10^5$, the regret ratio of our pricing mechanism with reserve price (\resp, with reserve price and uncertainty) is $7.77\%$ (\resp, $9.87\%$), reducing $57.19\%$ (\resp, $45.64\%$) of the regret ratio than a risk-averse baseline, where the reserve price is posted in each round; (2) for the pricing of accommodation rental under the log-linear model, when $n = 55$, $t = 74,111$, and the ratio between the natural logarithms of market value and reserve price is set to $0.6$, the regret ratio of our pricing mechanism is $3.83\%$, reducing $77.46\%$ of the regret ratio compared with the risk-averse baseline; (3) for the pricing of impression under the logistic model, when $n = 1024$ and $t = 10^5$, the regret ratios of our pure pricing mechanism are $8.04\%$ and $0.89\%$ in the spare and dense cases, respectively. Furthermore, the online latencies of three applications per round are in the magnitude of millisecond, and the memory overheads are less than 160MB.


$\bullet$~We instructively demonstrate that the reserve price can mitigate the cold-start problem in a posted price mechanism, and thus can reduce the cumulative regret. Specifically, for the pricing of noisy linear query, when $n = 20$ and $t = 10^4$, our pricing mechanism with reserve price (\resp, with reserve price and uncertainty) reduces $13.16\%$ (\resp, $10.92\%$) of the cumulative regret than without reserve price; for the pricing of accommodation rental, as the reserve price is approaching the market value, its impact on mitigating cold start is more evident. These findings may be of independent interest in the posted price mechanism design.



The remainder of this paper is organized as follows. In Section~\ref{sec:tech:overview}, we introduce technical preliminaries, and overview design principles. In Section~\ref{sec:linear:model}, we present our pricing mechanism under the fundamental linear model with uncertainty, and then analyze its performance from the time and space complexities together with the worst-case regret. In Section~\ref{sec:extensions}, we extend to non-linear models and some other similar application scenarios. We present the evaluation results in Section~\ref{sec:evaluations}. We briefly review related work in Section~\ref{sec:related:work}, and conclude the paper in Section~\ref{sec:conclusion}.

\section{Technical Overview}\label{sec:tech:overview}
In this section, we introduce system model, problem formulation, and design principles.

\subsection{System Model}

As shown in Fig.~\ref{fig:system:model}, we consider a general system model for online personal data markets. There are three kinds of entities: data owners, a data broker, and data consumers.

The data broker first collects massive personal data from data owners. Typical examples of personal data include product ratings, electrical usages, social media data, health records, physical activities, and driving trajectories. Then, the data consumers comes to the data market in an online fashion. In round $t \in [T]$, a data consumer arrives, and makes her customized query $Q_t$ over the collected dataset. Specifically, $Q_t$ comprises a concrete data analysis method and a tolerable level of noise added to the true result~\cite{jour:cacm:2017:roth,jour:cacm:2017:li}. Here, the noise perturbation can not only allow the data consumer to control the accuracy of a returned answer, but also preserve the privacies of data owners.


Depending on $Q_t$ and the underlying dataset, the data broker quantifies the privacy leakage of each data owner, and needs to compensate her if a deal occurs. Here, the individual privacy compensation, which hinges on the contract between the data owner and the data broker with respect to distinct privacy leakages and corresponding compensations, can be pre-computed when given $Q_t$. The data broker then offers a price $p_t$ to the data consumer. If $p_t$ is no more than the market value $v_t$ of $Q_t$, this posted price will be accepted. The data broker charges the data consumer $p_t$, returns the noisy answer, and compensates the data owners as planned. Otherwise, this deal is aborted, and the data consumer goes away. We note that to guarantee non-negative utility at the data broker no matter whether a deal occurs in round $t$ or not, the posted price $p_t$ should be no less than the total privacy compensation $q_t$, where $q_t$ functions as the {\em reserve price} of $Q_t$.


\subsection{Problem Formulation}\label{subsec:problem:formulation}

We now formulate the regret minimization problem for pricing sequential queries in online personal data markets.


We first model the market values of queries. We use an elementary assumption from \emph{contextual pricing} in computational economics~\cite{proc:ec16:cohen,proc:ec17:leme,proc:focs18:leme,proc:nips18:mao} and \emph{hedonic pricing} in marketing~\cite{jour:hedonic:milon1984,jour:hedonic:malpezzi2002,jour:hedonic:sirmans2005,proc:kdd18:airbnb}, which states that the market value of a product is a deterministic function of its features. Here, the product is a query, and the function can be linear or non-linear. Besides, to make the pricing model more robust, we allow for some uncertainty in the market value of each query. In particular, for a query $Q_t$, we let $\mathbf{x}_t \in \mathbb{R}^n$ denote its $n$-dimensional feature vector, let $f:\mathbb{R}^n \mapsto \mathbb{R}$ denote the mapping from the feature vector $\mathbf{x}_t$ to the deterministic part in its market value, and let $\delta_t \in \mathbb{R}$ denote the random variable in its market value, which is independent of $\mathbf{x}_t$. In a nutshell, $v_t = f(\mathbf{x}_t) + \delta_t$.


We next identify the features of a query for measuring its market value. One naive way is to directly encode the contents of the query, including the data analysis method and the noise level. However, the query alone, especially the data analysis method, is hard to embody its economic value. Thus, we turn to utilizing the underlying valuations from massive data owners about the query, namely the privacy compensations, as the feature vector. We give some comments on such a feature representation: (1) The market value of a query depending on the privacy compensations inherits the core principle of \emph{cost-plus pricing}~\cite{book:2003:pricing,book:2018:pricing}, and has been widely used in personal data pricing~\cite{jour:cacm:2017:li,proc:kdd:2018:erato,proc:infocom:2019:horae}. In particular, cost-plus pricing states that the market value of a product is determined by adding a specific amount of markup to its cost. Here, the cost is the total privacy compensation, the determinacy is reflected in the feature representation, and the markup is realized by setting the reserve price constraint. (2) The privacy compensations, incorporating the factors of the data analysis method, the noise level as well as the underlying dataset, are observable by the data broker, and can help her to discriminate the economic values of distinct queries. For example, the privacy compensations are higher, which implies that the privacy leakages to the data owners are larger, the knowledge discovered by the data consumer is richer, and thus the market value of the query to the data consumer should be higher. (3) Considering the large scale of data owners, the dimension of feature vector can be prohibitively high. Under such circumstance, we can apply some celebrated dimensionality reduction techniques, \eg, Principal Components Analysis (PCA)~\cite{book:06:bishop,book:09:esl}. Yet, we can also apply aggregation/clustering to the privacy compensations, and regard the aggregate results as the feature vector, where its dimension $n$ controls the granularity of aggregation. For example, we can sort the privacy compensations, and evenly divide them into $n$ partitions. We sum the privacy compensations falling into a certain partition, and thus obtain a feature. In this aggregation pattern, one extreme case is $n=1$, where the only feature is the total privacy compensation. Another extreme case is $n$ equal to the number of data owners, where every feature corresponds to a data owner's individual privacy compensation.





We finally define the cumulative regret of the data broker due to her limited knowledge of market values. We consider a game between the data broker and an adversary. During this game, the adversary chooses the sequence of queries $Q_1, Q_2, \ldots,  Q_T$, selects the mapping $f$, but cannot control the uncertainty $\delta_t$ in each round $t$, \ie, she can determine the part $f(\mathbf{x}_t)$ in the market value $v_t$. In contrast, the data broker can only passively receive each query $Q_t$, and then post a price $p_t$. If the posted price is no more than the market value, \ie, $p_t \leq v_t$, a deal occurs, and the data broker earns a revenue of $p_t$. Otherwise, the deal is aborted, and the data broker gains no revenue. We define the regret in round $t$ as the difference between the adversary's revenue and the data broker's revenue for trading the query $Q_t$, \ie,
\begin{equation*}
R_t = \left\{
\begin{aligned}
&0 \ &\text{if}\ q_t > v_t,\\
&\max_{{p_t^*}}{p_t^*}\ \underset{\delta_t}{\text{Pr}}\left({p_t^*} \leq v_t\right) - p_t\mathbf{1}\left\{p_t \leq v_t\right\}\ &\text{otherwise}.
\end{aligned}
\right.
\end{equation*}
Here, in the first branch, if the reserve price and thus the posted price are higher than the market value, there is no regret. This is because under such circumstance, no matter whether the adversary knows the market value in advance or the data broker does not, there is definitely no deal/revenue. Besides, $p_t^*$ is the adversary's optimal posted price to maximize her expected revenue in round $t$, where the expectation is taken over $\delta_t$. When $\delta_t$ is omitted, the adversary will just post the market value, if the reserve price is no more than the market value, \ie, $q_t \leq p_t^* = v_t$, and $R_t$ will change to:
\begin{equation}
R_t = \left\{
\begin{aligned}
&0 \ &\text{if}\ q_t > v_t,\\
&v_t - p_t \mathbf{1}\left\{p_t \leq v_t\right\}\ &\text{otherwise}.
\end{aligned}
\right.\label{eq:regret:formula:rp}
\end{equation}
At last, considering the queries can be chosen adversarially, \eg, by other competitive data brokers or malicious data consumers, our design goal is to minimize the total worst-case regret accumulated over $T$ rounds. Besides, good pricing policies under the metric of worst-case regret would still be robust to changes or fluctuations in the arrival pattern of queries over time, \eg, some features of the queries can appear in a correlated form, or have zero values throughout most of rounds, but may be important in later rounds.

%


\begin{figure*}[t]
\centering
\subfigure[One-Dimensional Scenario]{\label{fig:line:exploration}
\includegraphics[width=0.66\columnwidth]{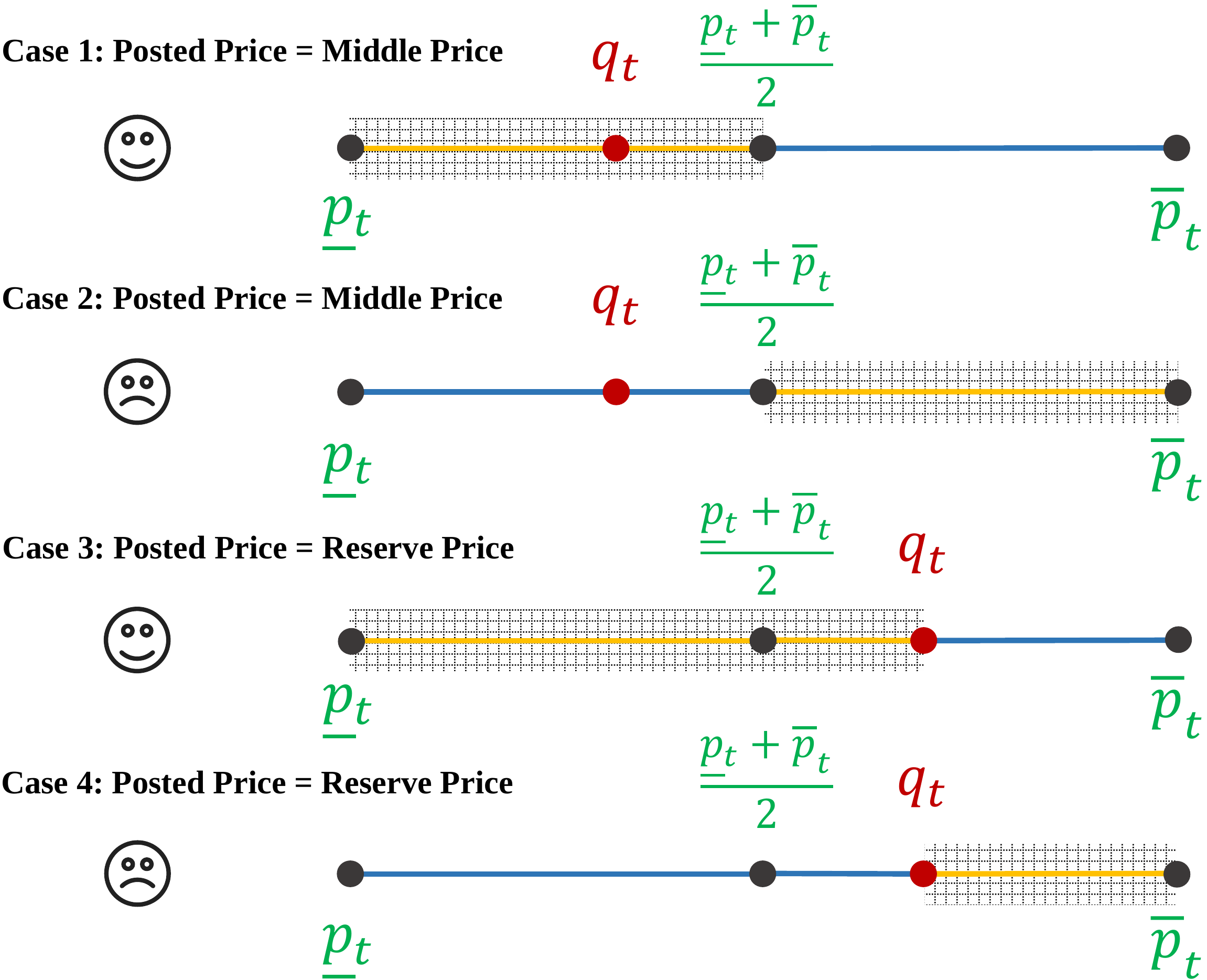}}
\subfigure[Multi-Dimensional Scenario]{\label{fig:ellipsoid:exploration}
\includegraphics[width=0.66\columnwidth]{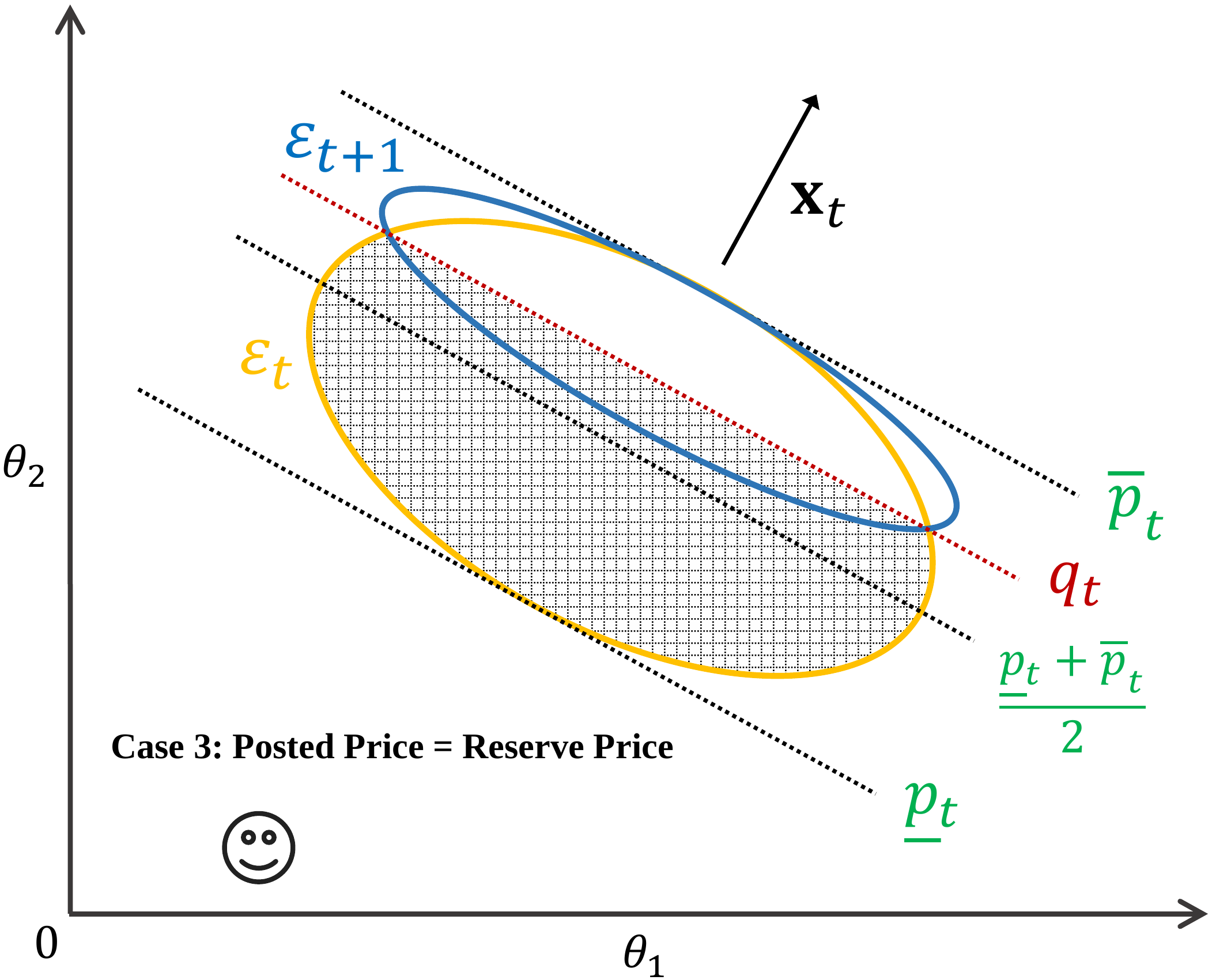}}
\subfigure[Multi-Dimensional Scenario with Uncertainty]{\label{fig:ellipsoid:exploration:noise}
\includegraphics[width=0.66\columnwidth]{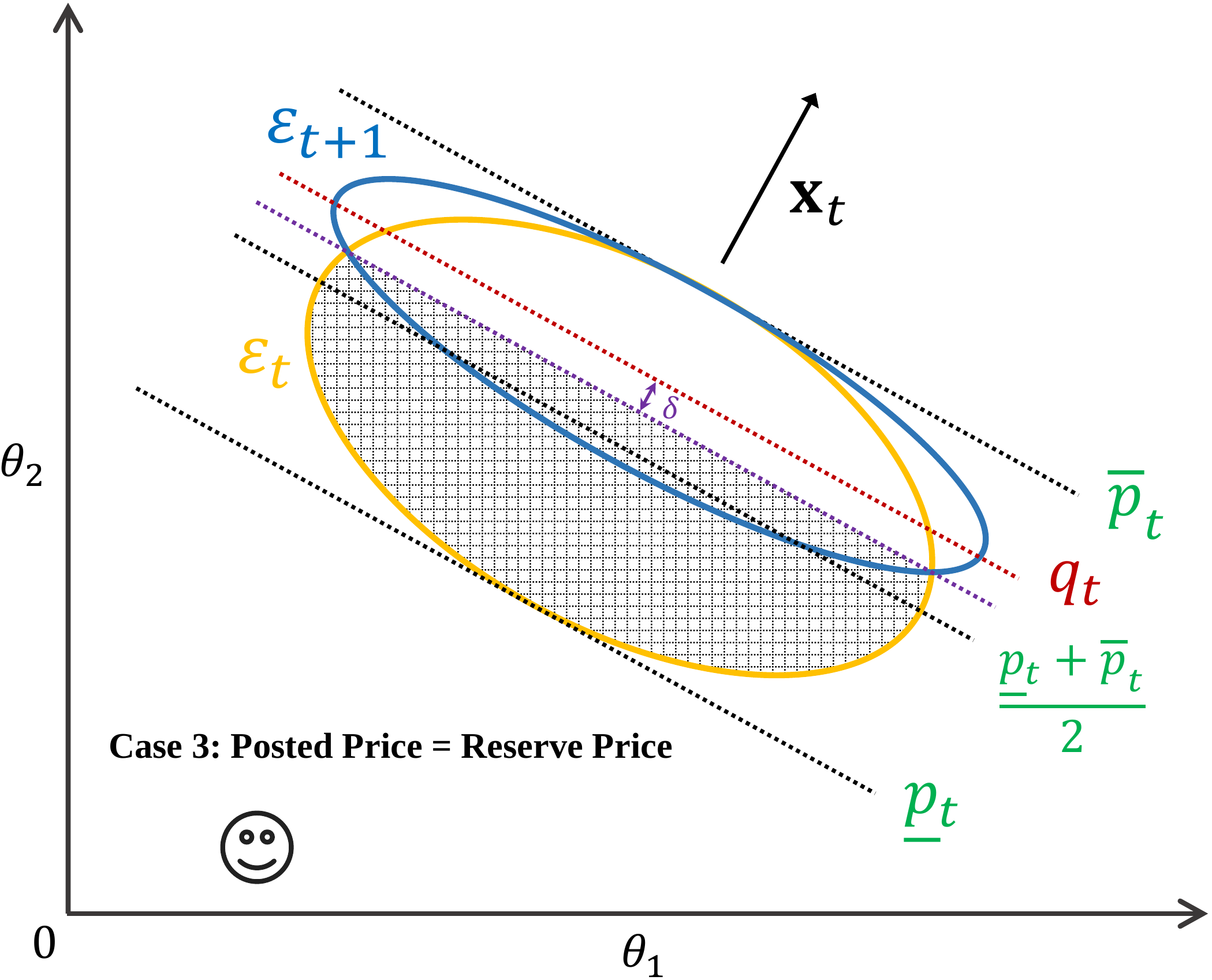}}
\caption{Illustrations of (effective) exploratory posted prices under the linear market value model.}
\vspace{-0.4em}
\end{figure*}

\subsection{Design Principles}\label{sec:design:principles}

We give an overview of our proposed pricing framework, and illustrate its key principles. We first consider the linear market value model, where $f$ is a linear function, parameterized by a weight vector $\theta^* \in \mathbb{R}^n$. In other words, the market value of the query $Q_t$ is $v_t = {\mathbf{x}_t}^T \theta^*$. We then consider extensions to the uncertain setting and non-linear models.




We start with a special case of the linear model, where each feature vector $\mathbf{x}_t$ is one-dimensional, \ie, $n=1$. For example, the single feature can be the total privacy compensation or the reserve price $q_t$, and the weight $\theta^*$ denotes some fixed but unknown revenue-to-cost ratio. We note that to minimize the regret in pricing the query $Q_t$, the data broker needs to have a good estimation of its market value $v_t$, which can be equivalently converted to having a good knowledge of $\theta^*$. We let $\mathcal{K}_t$ denote the data broker's knowledge set (or intuitively, feasible values) of $\theta^*$ in round $t$, \eg, the initial knowledge set $\mathcal{K}_1$ can be an interval $[\ell, u]$ for some $\ell, u \in \mathbb{R}$. Besides, after round $t$, if the posted price $p_t$ is rejected (\resp, accepted), the data broker will update her knowledge set $\mathcal{K}_t$ to $\mathcal{K}_{t+1} = \mathcal{K}_t \bigcap \{\theta \in \mathbb{R} | p_t \geq {\mathbf{x}_t}^T \theta\}$ (\resp, $\mathcal{K}_{t+1} = \mathcal{K}_t \bigcap \{\theta \in \mathbb{R} | p_t \leq {\mathbf{x}_t}^T \theta\}$). Now, the key problem for the data broker is how to set the posted price $p_t$. In fact, the knowledge set $\mathcal{K}_t$ can impose a lower bound $\ubar{p}_t = \min_{\theta \in \mathcal{K}_t} {\mathbf{x}_t}^T\theta$ and an upper bound $\bar{p}_t = \max_{\theta \in \mathcal{K}_t} {\mathbf{x}_t}^T\theta$ on estimating the market value $v_t$ and thus on the posted price $p_t$, while the reserve price $q_t$ imposes the other lower bound on the posted price $p_t$. If the posted price $p_t$ is $\max(q_t, \ubar{p}_t)$, the data broker can sell the query $Q_t$ with the highest probability. However, in the worst case, where $q_t \leq \ubar{p}_t$, this deal will not refine the knowledge set, \ie, $\mathcal{K}_{t+1} = \mathcal{K}_t$, and thus cannot benefit the following rounds. We call such a price $\max(q_t, \ubar{p}_t)$ a {\em conservative price}. On the other hand, as shown in Fig.~\ref{fig:line:exploration}, inspired by bisection, we define the larger value of the reserve price and the middle price, namely $\max(q_t, \frac{\ubar{p}_t + \bar{p}_t}{2})$, as an {\em exploratory price}. In the worst case, the feedback from posting this price can narrow down the knowledge set $\mathcal{K}_t$ by most, and thus can benefit the subsequent rounds most. Of course, compared with the conservative price, the exploratory price would suffer a higher risk of no sale or losing the current revenue. Besides, we note that both the conservative price and the exploratory price have adequately exploited the experience from the previous rounds, \ie, the latest knowledge set $\mathcal{K}_t$, while the difference is that these two types of posted prices give distinct biases to the immediate rewards (exploitation) and the future rewards (exploration). Yet, another key problem arisen is when the data broker should choose which price. Our strategy is to measure the size of the knowledge set $\mathcal{K}_t$, \eg, the width of interval in the one-dimensional case. If it exceeds some threshold, the data broker chooses the exploratory price to further improve her knowledge set; otherwise, her knowledge set is near optimal, and she chooses the conservative price. In our real design, we use $\bar{p}_t - \ubar{p}_t$ to capture the size of $\mathcal{K}_t$, and let $\epsilon > 0$ denote the threshold.


We further consider the general linear model with multiple features, \ie, $n \geq 2$. The holistic process is the same, whereas the difference lies in the concrete form of the knowledge set $\mathcal{K}_t$. In the one-dimensional case, $\mathcal{K}_t$ is an interval, and both the conservative price and the exploratory price can be efficiently computed from $\mathcal{K}_t$. However, when extended to the multi-dimensional case, we assume that the initial knowledge set is $\mathcal{K}_1 = \{\theta \in \mathbb{R}^n|\ell_i \leq \theta_i \leq u_i, \ell_i, u_i \in \mathbb{R}\}$. Besides, after each round, the knowledge set is updated by adding a linear inequality. Thus, the knowledge set $\mathcal{K}_t$ can be viewed as a set of linear inequalities, whose cardinality is non-decreasing with the number of rounds $t$. To post a price in round $t$, it suffices to solve two linear programs under $\mathcal{K}_t$, which is quite time-consuming, and can be computationally infeasible in online mode. Therefore, we turn to borrowing some key principles from the celebrated ellipsoid method for solving online linear programs, which was first proposed by Khachiyan in 1979~\cite{jour:khachiyan:1979}. The pivotal idea is to replace the raw knowledge set $\mathcal{K}_t$, viewed as a polytope in geometry, with the ellipsoid $\mathcal{E}_t$ of minimum volume that contains $\mathcal{K}_t$. This ellipsoid is normally referred to as the L\"{o}wner-John ellipsoid of the convex body $\mathcal{K}_t$. In addition, by leveraging the property that every ellipsoid is an image of the unit ball under a bijective affine transformation~\cite{book:ellipsoid:depth:cut}, the data broker can efficiently determine each posted price, and further update her knowledge set, which only consume a few matrix-vector and vector-vector multiplications. Fig.~\ref{fig:ellipsoid:exploration} gives an illustration of the exploratory posted price in the two-dimensional case. Here, the points on the dashed line, perpendicular to the feature vector $\mathbf{x}_t$ and tagged with a certain price, represent those weight vectors to derive this price. Besides, in round $t$, the data broker's knowledge about the weight vector $\theta^*$ comprises all the points in the yellow ellipsoid $\mathcal{E}_t$. Based on $\mathcal{E}_t$, the data broker can compute a lower bound $\ubar{p}_t$ and an upper bound $\bar{p}_t$ on the market value $v_t$. Given $\bar{p}_t - \ubar{p}_t$ exceeds the threshold $\epsilon$, the data broker should post the exploratory price $\max(q_t, \frac{\ubar{p}_t + \bar{p}_t}{2})$, \ie, the reserve price $q_t$ here. In addition, this posted price is accepted by the data consumer, which implies that the posted price is no more than the market value, \ie, $p_t = q_t \leq v_t = {\mathbf{x}_t}^T\theta^*$. This feedback enables the data broker to refine her knowledge set $\mathcal{E}_t$, by excluding the region below the cutting line $\{\theta \in \mathbb{R}^2| p_t = {\mathbf{x}_t}^T\theta\}$, which is marked with grid. Furthermore, the new knowledge set $\mathcal{E}_{t+1}$ is obtained by replacing the remaining region of $\mathcal{E}_t$ with its L\"{o}wner-John ellipsoid, \ie, all the points in the blue ellipsoid.



We finally investigate the extensions to the uncertain setting and non-linear models. First, for tractability, we make a common assumption on the randomness $\delta_t$ in the market value $v_t$, where the distribution of $\delta_t$ belongs to subGaussian. We thus bound the absolute value of any $\delta_t$ in all $T$ rounds by $\delta$, with probability near $1$. We regard $\delta$ as a ``buffer" in posting the price and updating the knowledge set, which can circumvent the randomness $\delta_t$ in each round. Second, we mainly investigate four classic non-linear models in market value estimation, whose pattern is first applying an inner feature mapping to the feature vector, then performing dot product with the weight vector, and finally applying an outer non-decreasing and continuous function. By still focusing on the discovery of the weight vector rather than the inner and outer non-linear functions, we can extend our pricing mechanism to this class of non-linear market value models.



\section{Fundamental Design Under Linear Market Value Model}\label{sec:linear:model}

In this section, we first propose the ellipsoid based pricing mechanism under the linear market value model, and then extend to tolerate some uncertainty. We further analyze the performance from the time and space complexities together with the worst-case regret.

\subsection{Ellipsoid based Pricing Mechanism}\label{subsec:version:reserve:price}

\begin{algorithm}[!t]
\small	
\caption{Ellipsoid based Personal Data Pricing}\label{alg:ellipsoid:pricing}

\KwIn{$\mathbf{A}_1 = R^2\mathbf{I}_{n \times n}$, $\mathbf{c}_1 = \mathbf{0}_{n \times 1}$, a threshold $\epsilon$}
\KwOut{Posted price $p_t$ in each round $t \in [T]$}
\For{$t = 1, 2, \ldots, T$}
{
   $\mathcal{E}_{t} = \{\theta \in \mathbb{R}^n | \left(\theta - \mathbf{c}_{t}\right)^T \mathbf{A}_{t}^{-1}\left(\theta - \mathbf{c}_{t}\right)\leq 1\}$\;
   Receive a query $Q_t$ with the feature vector $\mathbf{x}_t \in \mathbb{R}^n$\;
   Determine the reserve price $q_t$ of $Q_t$\;
   $\mathbf{b}_t = \frac{\mathbf{A}_t \mathbf{x}_t}{\sqrt{{\mathbf{x}_t}^T\mathbf{A}_t \mathbf{x}_t}}$\;
   $\ubar{p}_t = \min_{\theta \in \mathcal{E}_t} {\mathbf{x}_t}^T\theta = {\mathbf{x}_t}^T\left(\mathbf{c}_t - \mathbf{b}_t\right)$\;
   $\bar{p}_t = \max_{\theta \in \mathcal{E}_t} {\mathbf{x}_t}^T \theta = {\mathbf{x}_t}^T\left(\mathbf{c}_t + \mathbf{b}_t\right)$\;
   \If{$q_t \geq \bar{p}_t$}
   {
      $\mathbf{A}_{t+1} = \mathbf{A}_t; \mathbf{c}_{t+1} = \mathbf{c}_t$\;
      \Continue\;
   }
   \Else
   {
     \If{$\bar{p}_t - \ubar{p}_t = 2 \sqrt{{\mathbf{x}_t}^T\mathbf{A}_t \mathbf{x}_t} > \epsilon$}
     {
        Post the price $p_t = \max\left\{q_t, \frac{\ubar{p}_t + \bar{p}_t}{2} = {\mathbf{x}_t}^T\mathbf{c}_t\right\}$\;

        $\alpha_t = \frac{\frac{\ubar{p}_t + \bar{p}_t}{2} - p_t}{\sqrt{{\mathbf{x}_t}^T \mathbf{A}_t \mathbf{x}_t}} = \frac{{\mathbf{x}_t}^T\mathbf{c}_t - p_t}{\sqrt{{\mathbf{x}_t}^T \mathbf{A}_t \mathbf{x}_t}}$\;
        \If{$p_t$ {\em is rejected}}
        {
          \If{$-\frac{1}{n} \leq \alpha_t \leq 0$}
          {
            $
             \begin{multlined}[c]
                 \mathbf{A}_{t+1} = \frac{n^2\left(1 - {\alpha_t}^2\right)}{n^2 - 1}\left(\vphantom{- \frac{2\left(1 + n \alpha_t\right)}{\left(n + 1\right)\left(1 + \alpha_t\right)} {\mathbf{b}_t}{\mathbf{b}_t}^T}\right. \mathbf{A}_t\\
                   \quad\quad\ \left. - \frac{2\left(1 + n \alpha_t\right)}{\left(n + 1\right)\left(1 + \alpha_t\right)} {\mathbf{b}_t}{\mathbf{b}_t}^T\right);
             \end{multlined}
            $
           $\mathbf{c}_{t+1} = \mathbf{c}_t - \frac{1 + n\alpha_t}{n + 1} \mathbf{b}_t$\;
         }
         \Else
         {
           $\mathbf{A}_{t+1} = \mathbf{A}_t; \mathbf{c}_{t+1} = \mathbf{c}_t$\;
         }
     }
     \Else
     {
              $
                 \begin{multlined}[c]
                    \mathbf{A}_{t+1} = \frac{n^2\left(1 - {\alpha_t}^2\right)}{n^2 - 1}\left(\vphantom{- \frac{2\left(1 - n \alpha_t\right)}{\left(n + 1\right)\left(1 - \alpha_t\right)} {\mathbf{b}_t}{\mathbf{b}_t}^T}\right.\mathbf{A}_t\\
                    \quad\quad\ \left. - \frac{2\left(1 - n \alpha_t\right)}{\left(n + 1\right)\left(1 - \alpha_t\right)} {\mathbf{b}_t}{\mathbf{b}_t}^T\right);
                 \end{multlined}
              $
           $\mathbf{c}_{t+1} = \mathbf{c}_t + \frac{1 - n\alpha_t}{n + 1} \mathbf{b}_t$\;
     }
     }
     \Else
     {
        Post the price $p_t = \max\left\{q_t, \ubar{p}_t\right\}$\;
        $\mathbf{A}_{t+1} = \mathbf{A}_t; \mathbf{c}_{t+1} = \mathbf{c}_t$;
     }
   }
}
\end{algorithm}

As an appetizer, we first briefly review the definition of an ellipsoid and some of its key properties.
\begin{definition}\label{def:ellipsoid}
$\mathcal{E} \subseteq \mathbb{R}^n$ is an ellipsoid, if there exists a vector $\mathbf{c} \in \mathbb{R}^n$ and a positive definite matrix $\mathbf{A} \in \mathbb{R}^{n \times n}$ such that:
\begin{align}
\mathcal{E} = \left\{\theta \in \mathbb{R}^n \middle| \left(\theta - \mathbf{c}\right)^T \mathbf{A}^{-1}\left(\theta - \mathbf{c}\right) \leq 1 \right\}.\label{eq:ellipsoid:format}
\end{align}
\end{definition}

Intuitively, $\mathbf{c}$ represents the center of the ellipsoid $\mathcal{E}$, while $\mathbf{A}$ portrays its shape. In particular, there are some useful connections between the geometric properties of $\mathcal{E}$ and the algebraic properties of $\mathbf{A}$. We let $\gamma_i(\mathbf{A}) > 0$ denote the $i$-th largest eigenvalue of $\mathbf{A}$. Then, the $i$-th widest axis (\resp, its width) of the ellipsoid $\mathcal{E}$ corresponds to the $i$-th eigenvector (\resp, $2\sqrt{\gamma_i(\mathbf{A)}}$). Besides, the volume of the ellipsoid $\mathcal{E}$, denoted as $V(\mathcal{E})$, only depends on the eigenvalues of $\mathbf{A}$ and the dimension $n$. Specifically,
\begin{align}
V\left(\mathcal{E}\right) = V_n \sqrt{\prod_{i\in\left[n\right]} \gamma_i\left(\mathbf{A}\right)},\label{eq:ellipsoid:volume}
\end{align}
where $V_n$ denotes the volume of the unit ball in $\mathbb{R}^n$, and is a constant that only hinges on $n$.


We now present the ellipsoid based posted price mechanism with the reserve price constraint for online personal data markets in Algorithm~\ref{alg:ellipsoid:pricing} (also called ``the version with reserve price" in our evaluation part). We recall that the initial knowledge set of the data broker about the weight vector $\theta^*$ is $\mathcal{K}_1 = \{\theta \in \mathbb{R}^n|\ell_i \leq \theta_i \leq u_i, \ell_i, u_i \in \mathbb{R}\}$. We choose a ball centered at the origin with radius $R = \sqrt{\sum_{i\in[n]} \max({\ell_i}^2, {u_i}^2)}$ to enclose $\mathcal{K}_1$. This ball can serve as the initial ellipsoid $\mathcal{E}_1$, where $\mathbf{A}_1 = R^2 \mathbf{I}_{n\times n}$ and $\mathbf{c}_1 = \mathbf{0}_{n \times 1}$. In what follows, we focus on a concrete round $t$. The data broker receives a query $Q_t$ with the feature vector $\mathbf{x}_t$ from a data consumer. Without loss of generality, we assume that $\forall t \in [T], \|\mathbf{x}_t\| \leq S$ for some $S \geq 1$. Then, she virtually computes the total privacy compensation allocated to the data owners as the reserve price $q_t$, which imposes a strict lower bound on the posted price $p_t$. Based on the knowledge set $\mathcal{E}_t$, the data broker can elicit that the market value of the query $Q_t$ falls into a certain interval, \ie, $v_t = {\mathbf{x}_t}^T\theta^* \in [\ubar{p}_t, \bar{p}_t]$ (Lines 5--7). If the reserve price is no less than the maximum possible market value, implying that the posted price should be no less than the market value, \ie, $p_t \geq q_t \geq \bar{p}_t \geq v_t$, the query $Q_t$ cannot be sold (Lines 8--10). Otherwise, the data broker judges whether the difference between the maximum and minimum possible market values, \ie, $\bar{p}_t - \ubar{p}_t$, exceeds a threshold $\epsilon$. If yes, she posts the exploratory price (Lines 12--13). Otherwise, she posts the conservative price (Lines 22--23). In fact, the posted price places a cut on the ellipsoid $\mathcal{E}_t$, and splits it into two parts, where the cutting hyperplane is $\{\theta \in \mathbb{R}^n| p_t = {\mathbf{x}_t}^T\theta\}$. Besides, the data broker can compute a parameter $\alpha_t$ to locate the position of the cut (Line 14). Formally, $\alpha_t$ is interpreted as the signed distance from the center $\mathbf{c}_t$ to the cutting hyperplane, measured in the space $\mathbb{R}^n$ endowed with the ellipsoidal norm $\|\cdot\|_{{\mathbf{A}_t}^{-1}}$~\cite{jour:or81:bland,book:ellipsoid:depth:cut}. For example, if the posted price is the middle price, \ie, $p_t = \frac{\ubar{p}_t + \bar{p}_t}{2} = {\mathbf{x}_t}^T \mathbf{c}_t$, the center $\mathbf{c}_t$ is on the cutting hyperplane, and thus $\alpha_t = 0$. Moreover, according to the feedback from the data consumer, the data broker can decide to retain which side of the ellipsoid $\mathcal{E}_t$, and thus update to its L\"{o}wner-John ellipsoid $\mathcal{E}_{t+1}$, by computing the new shape $\mathbf{A}_{t+1}$ and center $\mathbf{c}_{t+1}$ (Lines 15--21). In particular, Gr{\"o}tschel~\et~\cite{book:ellipsoid:depth:cut} have offered the formulas of $\mathbf{A}_{t+1}$ and $\mathbf{c}_{t+1}$, when the remaining part of $\mathcal{E}_t$ is contained in the halfspace like $\{\theta \in \mathbb{R}^n |p_t \geq {\mathbf{x}_t}^T\theta\}$. This corresponds to the rejection branch (Lines 15--19). By the symmetry of ellipsoid, we can obtain the formulas in the acceptance branch (Lines 20--21). Furthermore, if the remaining part after a cut is exactly half of the ellipsoid $\mathcal{E}_t$, we call the cut a {\em central cut}, if the remaining part is less than half, we call it a {\em deep cut}, and if the remaining part is more than half, we call it a {\em shallow cut}. Last but not least, it is worth noting that the data broker is prohibited from refining the ellipsoid with the conservative price (Line 24). The reason is that $\bar{p}_t - \ubar{p}_t$ essentially probes the ellipsoid's width along the direction given by the feature vector $\mathbf{x}_t$ (Please see Fig.~\ref{fig:ellipsoid:exploration} for an intuition.), which is very small ($\leq \epsilon$) when posting the conservative price. Suppose the data broker is allowed to cut along this direction. By adversarially setting the reserve prices, the width of ellipsoid along this direction can shrink successively, while the widthes along other directions can expand exponentially, which can result in $O(T)$ worst-case regret. Details about the adversarial example and its regret analysis are deferred to Lemma~\ref{lem:adversarial:example} in Appendix.



We finally discuss a special case, where Algorithm~\ref{alg:ellipsoid:pricing} executes without the reserve price constraint, hereinafter denoted as Algorithm~\ref{alg:ellipsoid:pricing}* (also called ``the pure version" in our evaluation part). First, the exploratory posted price takes the middle price $\frac{\ubar{p}_t + \bar{p}_t}{2}$, and poses a central cut over the ellipsoid $\mathcal{E}_t$. Second, the conservative posted price takes $\ubar{p}_t$, which is definitely no more than the market value $v_t$, and must be accepted by the data consumer. Besides, the conservative price does not refine $\mathcal{E}_t$, and thus incurs a shallow cut. In a nutshell, there is no deep cut in Algorithm~\ref{alg:ellipsoid:pricing}*.

\subsection{Incorporating Uncertainty}

\begin{algorithm}[!t]
\small
\caption{Personal Data Pricing with Uncertainty}\label{alg:ellipsoid:pricing:noise}

\KwIn{$\mathbf{A}_1 = R^2\mathbf{I}_{n \times n}$, $\mathbf{c}_1 = \mathbf{0}_{n \times 1}$, an uncertainty parameter $\delta = \sqrt{2\log C}\sigma \log T$, a threshold $\epsilon$}
\KwOut{Posted price $p_t$ in each round $t \in [T]$}
\For{$t = 1, 2, \ldots, T$}
{
   $\mathcal{E}_{t} = \{\theta \in \mathbb{R}^n | \left(\theta - \mathbf{c}_{t}\right)^T \mathbf{A}_{t}^{-1}\left(\theta - \mathbf{c}_{t}\right)\leq 1\}$\;
   Receive a query $Q_t$ with the feature vector $\mathbf{x}_t \in \mathbb{R}^n$\;
   Determine the reserve price $q_t$ of $Q_t$\;
   $\mathbf{b}_t = \frac{\mathbf{A}_t \mathbf{x}_t}{\sqrt{{\mathbf{x}_t}^T\mathbf{A}_t \mathbf{x}_t}}$\;
   $\ubar{p}_t = \min_{\theta \in \mathcal{E}_t} {\mathbf{x}_t}^T\theta = {\mathbf{x}_t}^T\left(\mathbf{c}_t - \mathbf{b}_t\right)$\;

   $\bar{p}_t = \max_{\theta \in \mathcal{E}_t} {\mathbf{x}_t}^T \theta = {\mathbf{x}_t}^T\left(\mathbf{c}_t + \mathbf{b}_t\right)$\;
   \If{$q_t \geq \bar{p}_t + \delta$}
   {
      $\mathbf{A}_{t+1} = \mathbf{A}_t$; $\mathbf{c}_{t+1} = \mathbf{c}_t$\;
      \Continue\;
   }
   \Else
   {
     \If{$\bar{p}_t - \ubar{p}_t = 2 \sqrt{{\mathbf{x}_t}^T\mathbf{A}_t \mathbf{x}_t} > \epsilon$}
     {
       Post the price $p_t = \max\left\{q_t, \frac{\ubar{p}_t + \bar{p}_t}{2} = {\mathbf{x}_t}^T\mathbf{c}_t\right\}$\;
       \If{$p_t$ {\em is rejected}}
       {
         $\alpha_t = \frac{\frac{\ubar{p}_t + \bar{p}_t}{2} - \left(p_t + \delta\right)}{\sqrt{{\mathbf{x}_t}^T \mathbf{A}_t \mathbf{x}_t}} = \frac{{\mathbf{x}_t}^T\mathbf{c}_t - p_t - \delta}{\sqrt{{\mathbf{x}_t}^T \mathbf{A}_t \mathbf{x}_t}}$\;
         \If{$-\frac{1}{n} \leq \alpha_t \leq 1$}
         {
           $
             \begin{multlined}[c]
                 \mathbf{A}_{t+1} = \frac{n^2\left(1 - {\alpha_t}^2\right)}{n^2 - 1}\left(\vphantom{- \frac{2\left(1 + n \alpha_t\right)}{\left(n + 1\right)\left(1 + \alpha_t\right)} {\mathbf{b}_t}{\mathbf{b}_t}^T}\right. \mathbf{A}_t\\
                   \quad\quad\ \left. - \frac{2\left(1 + n \alpha_t\right)}{\left(n + 1\right)\left(1 + \alpha_t\right)} {\mathbf{b}_t}{\mathbf{b}_t}^T\right);
             \end{multlined}
           $
           $\mathbf{c}_{t+1} = \mathbf{c}_t - \frac{1 + n\alpha_t}{n + 1} \mathbf{b}_t$\;
         }
         \Else
         {
           $\mathbf{A}_{t+1} = \mathbf{A}_t$; $\mathbf{c}_{t+1} = \mathbf{c}_t$\;
         }
      }
      \Else
      {
         $\alpha_t = \frac{\frac{\ubar{p}_t + \bar{p}_t}{2} - \left(p_t - \delta\right)}{\sqrt{{\mathbf{x}_t}^T \mathbf{A}_t \mathbf{x}_t}} = \frac{{\mathbf{x}_t}^T\mathbf{c}_t - p_t + \delta}{\sqrt{{\mathbf{x}_t}^T \mathbf{A}_t \mathbf{x}_t}}$\;
         \If{$-\frac{1}{n} \leq -\alpha_t \leq 1$}
         {
            $
               \begin{multlined}[c]
                   \mathbf{A}_{t+1} = \frac{n^2\left(1 - {\alpha_t}^2\right)}{n^2 - 1}\left(\vphantom{- \frac{2\left(1 - n \alpha_t\right)}{\left(n + 1\right)\left(1 - \alpha_t\right)} {\mathbf{b}_t}{\mathbf{b}_t}^T}\right.\mathbf{A}_t\\
                   \quad\quad\ \left. - \frac{2\left(1 - n \alpha_t\right)}{\left(n + 1\right)\left(1 - \alpha_t\right)} {\mathbf{b}_t}{\mathbf{b}_t}^T\right);
               \end{multlined}
            $
           $\mathbf{c}_{t+1} = \mathbf{c}_t + \frac{1 - n\alpha_t}{n + 1} \mathbf{b}_t$\;
          }
          \Else
          {
            $\mathbf{A}_{t+1} = \mathbf{A}_t$; $\mathbf{c}_{t+1} = \mathbf{c}_t$\;
          }
       }
     }
     \Else
     {
       Post the price $p_t = \max\left\{q_t, \ubar{p}_t - \delta\right\}$\;
       $\mathbf{A}_{t+1} = \mathbf{A}_t$; $\mathbf{c}_{t+1} = \mathbf{c}_t$\;
     }
   }
}
\end{algorithm}


We now extend to the uncertain setting. We first make an assumption on the random variable $\delta_t$ in the market value model. We assume that the distribution of $\delta_t$ is $\sigma$-subGaussian, \ie, there exists a constant $C \in \mathbb{R}$ such that
\begin{align}
\forall z > 0, \text{Pr}\left(|\delta_t| > z\right) \leq C \exp\left(-\frac{z^2}{2\sigma^2}\right).\label{eq:sub:gaussian}
\end{align}
This is a common assumption widely used in modeling uncertainty~\cite{proc:nips09:subgaussian:C,proc:nips10:subgaussian,proc:icml13:chen:subgaussian,jour:vldb17:subgaussian}. In particular, many celebrated probability distributions, including normal distribution, uniform distribution, Rademacher distribution, and bounded random variables, are subGaussian. For example, normal distribution is $\sigma$-subGaussian for its standard deviation $\sigma$ and for $C = 2$. We then assign a value $\delta = \sqrt{2\log C} \sigma \log T$ to the variable $z$ in Equation~(\ref{eq:sub:gaussian}), and thus get:
\begin{align}
\text{Pr}\left(|\delta_t| > \delta\right) \leq T^{-\log T}.
\end{align}
We further apply Boole's inequality to the above inequality for all $t \in [T]$, and derive:
\begin{equation}
\begin{aligned}
&\exists t \in [T], \text{Pr}\left(|\delta_t| > \delta\right) \leq  T^{1-\log T}\\
\Rightarrow\quad&\forall t \in [T], \text{Pr}\left(|\delta_t| \leq \delta\right) \geq 1 - T^{1-\log T} \geq 1 - 1/T,
\end{aligned}\label{eq:noise:bound}
\end{equation}
where the last inequality holds for $T \geq 8$.

From Equation~(\ref{eq:noise:bound}), we can draw that in each round $t$, the randomness $\delta_t$ in the market value $v_t$ is bounded by $\delta$ in absolute value with probability at least $1 - 1/T$. Therefore, when posting the price and updating the knowledge set, we let the data broker introduce a ``buffer" of size $\delta$ to circumvent the randomness $\delta_t$. Specifically, if the data broker posts the price $p_t$ and observes a rejection, she can no longer infer that $p_t \geq {\mathbf{x}_t}^T \theta^*$. Instead, she should infer that $p_t \geq v_t = {\mathbf{x}_t}^T\theta^* - \delta_t \geq {\mathbf{x}_t}^T \theta^* - \delta$. In a similar way, if she observes an acceptance, she will infer that $p_t \leq v_t = {\mathbf{x}_t}^T\theta^* + \delta_t \leq {\mathbf{x}_t}^T\theta^* + \delta$ rather than $p_t \leq {\mathbf{x}_t}^T \theta^*$. Intuitively, in the case of rejection (\resp, acceptance), the data broker imagines that she had posted $p_t + \delta$ (\resp, $p_t - \delta$). We call $p_t + \delta$ (\resp, $p_t - \delta$) the {\em effective posted price} in the case of rejection (\resp, acceptance).

We now present the robust pricing mechanism in Algorithm~\ref{alg:ellipsoid:pricing:noise} (also called ``the version with reserve price and uncertainty" in our evaluation part). For the sake of conciseness, we mainly illustrate the differences from Algorithm~\ref{alg:ellipsoid:pricing}. First, in Lines 8--10, the condition for a certain no deal changes into $q_t \geq \bar{p}_t + \delta$. Only under this condition, the posted price must be no less than the market value, since $p_t \geq q_t \geq \bar{p}_t + \delta \geq v_t = {\mathbf{x}_t}^T\theta^* + \delta_t$. Second, in Lines 15 and 21, we utilize the effective exploratory prices to compute the positions of the cutting hyperplanes. In particular, due to the uncertainty in the market value, if the data broker posts the same price, the feedback from the data consumer can result in a smaller refinement of the knowledge set, \ie, the depth of the cut over the ellipsoid decreases. We provide Fig.~\ref{fig:ellipsoid:exploration:noise} for a visual comparison with Fig.~\ref{fig:ellipsoid:exploration}, where Fig.~\ref{fig:ellipsoid:exploration:noise} incorporates the uncertainty, while Fig.~\ref{fig:ellipsoid:exploration} does not. Third, in Line 27, the conservative posted price, involving $\ubar{p}_t$, decreases by $\delta$ to keep its high acceptance ratio.


We finally investigate Algorithm~\ref{alg:ellipsoid:pricing:noise} without the reserve price constraint, hereinafter denoted as Algorithm~\ref{alg:ellipsoid:pricing:noise}* (also called ``the version with uncertainty" in our evaluation part). First, the exploratory posted price is still the middle price $\frac{\ubar{p}_t + \bar{p}_t}{2}$. Besides, the effective exploratory price used in refining the ellipsoid is $\frac{\ubar{p}_t + \bar{p}_t}{2} + \delta$ (\resp, $\frac{\ubar{p}_t + \bar{p}_t}{2} - \delta$) in the case of rejection (\resp, acceptance), and the corresponding position parameter $\alpha_t$ is $-\delta/\sqrt{{\mathbf{x}_t}^T\mathbf{A}_t\mathbf{x}_t}$ (\resp, $\delta/\sqrt{{\mathbf{x}_t}^T\mathbf{A}_t\mathbf{x}_t}$). If $\delta > 0$, the effective exploratory prices will refine the ellipsoid less than half. Furthermore, the new ellipsoids $\mathcal{E}_{t+1}$'s obtained in the cases of acceptance and rejection are symmetric according to the hyperplane through the center of the old ellipsoid $\mathcal{E}_{t}$, namely $\{\theta \in \mathbb{R}^n| {\mathbf{x}_t}^T\theta = {\mathbf{x}_t}^T\mathbf{c}_t\}$. Second, the conservative posted price is $\ubar{p}_t - \delta$, and can be either rejected or accepted. Here, the rejection case happens when the market value is outside the interval $[\ubar{p}_t - \delta, \bar{p}_t + \delta]$, and has probability no more than $1/T$ from Equation~(\ref{eq:noise:bound}). Besides, the conservative price keeps the ellipsoid unchanged. Jointly considering two types of posted prices, we can find that if the uncertainty in the market value is incorporated, \ie, $\delta > 0$, Algorithm~\ref{alg:ellipsoid:pricing:noise}* only has shallow cuts. Of course, when the uncertainty is ignored by setting $\delta = 0$, the above analysis degenerates to the analysis of Algorithm~\ref{alg:ellipsoid:pricing}* at the end of Section~\ref{subsec:version:reserve:price}.





\subsection{Performance Analysis}
We analyze the performance of Algorithm~\ref{alg:ellipsoid:pricing:noise} from its time and space complexities as well as worst-case regret.

\subsubsection{Time and Space Complexities}
Considering the data broker needs to run the posted price mechanism in online mode, Algorithm~\ref{alg:ellipsoid:pricing:noise} should be quite efficient in each iteration. We analyze its single-round complexity from the computation and memory overheads.

First, the computation overhead of the data broker in a round of data trading mainly comes from two parts: one is to determine the posted price $p_t$, which roughly consumes 2 matrix-vector and 3 vector-vector multiplications; the other is to update the shape and the center of the ellipsoid, which roughly consumes 1 vector-vector multiplication in the worst case. Thus, the time complexity is $O(n^2)$. Second, the memory overhead of the data broker is mainly caused by maintaining the knowledge set $\mathcal{E}_t$, or alteratively the shape and the center of the ellipsoid, which requires one $n \times n$ matrix and one $n \times 1$ vector, respectively. Hence, the space complexity is $O(n^2)$.

In conclusion, our proposed pricing mechanism has a light load, and can apply to online personal data markets. We shall report the detailed overheads in Section~\ref{subsec:implementation:overhead}.

\subsubsection{Worst-Case Regret}\label{sec:regret:analysis}

We analyze the worst-case regret of Algorithm~\ref{alg:ellipsoid:pricing:noise}, which is $O(\max(n^2\log(T/n), n^3\log(T/n)/T))$ under the low uncertain setting $\delta = O(n/T)$, namely Theorem~\ref{theorem:alg2}. We first prove that the existence of reserve price cannot increase the regret of a posted price mechanism in single round (Lemma~\ref{lem:reserve:price:regret}). Thus, we can use Algorithm~\ref{alg:ellipsoid:pricing:noise} without the reserve price constraint, namely Algorithm~\ref{alg:ellipsoid:pricing:noise}*, as a springboard. In particular, to get an upper bound on the cumulative regret of Algorithm~\ref{alg:ellipsoid:pricing:noise}, we need to derive an upper bound on the number of rounds where the exploratory prices are posted. We derive this upper bound in a roundabout way: we first obtain the upper bound in Algorithm~\ref{alg:ellipsoid:pricing:noise}* (Lemma~\ref{lem:exploratory:number:round:alpha}), and further prove that it still holds in Algorithm~\ref{alg:ellipsoid:pricing:noise} by reduction and analyzing the impact of reserve price (Lemma~\ref{lem:original:exploratory:number:round}). We elicit Lemma~\ref{lem:exploratory:number:round:alpha} in a squeezing manner, particularly through constructing an upper bound and a lower bound on the final volume of the ellipsoid. For the upper bound, we adopt a core technique in proving the convergence of the celebrated ellipsoid method: the ratio between the volumes of an ellipsoid and the L\"{o}wner-John ellipsoid after a cut has an upper bound (Lemma~\ref{lem:ellipsoid:volume:ratio:alpha}). Regarding the lower bound, we resort to the formula for computing an ellipsoid's volume by multiplying all the eigenvalues of its shape matrix. Thus, we can find a lower bound on the volume, by constructing a lower bound on the smallest eigenvalue (Lemmas~\ref{lem:smallest:eigenvalue:nondecresing:alpha} and~\ref{lem:smallest:eigenvalue:change:alpha}). The cornerstone of these two lemmas is the linear algebra machinery in Lemma~\ref{lem:smallest:eigenvalue}, which captures how the smallest eigenvalue changes after a matrix is modified by a rank-one matrix. In what follows, we present, prove, and clarify the detailed lemmas and theorem in a bottom-up way.



\begin{lemma}\label{lem:reserve:price:regret}
The existence of reserve price cannot increase the regret of a posted price mechanism in single round.
\end{lemma}
\begin{proof}
We let $v_t$ denote the market value, let $q_t$ denote the reserve price, let $p_t'$ denote the pure posted price, and let $p_t$ denote the posted price with the reserve price constraint, \ie, $p_t = \max(q_t, p_t')$. We can express the regret of the posted price mechanism without reserve price in round $t$ as:
\begin{equation}
R_t' = v_t - p_t' \mathbf{1}\left\{p_t' \leq v_t\right\}.
\end{equation}
After introducing the reserve price constraint, the regret changes to $R_t$ given in Equation~(\ref{eq:regret:formula:rp}). We now prove $R_t \leq R_t'$ in two complementary cases: $q_t > v_t$ and $q_t \leq v_t$.

Case 1 ($q_t > v_t$): We can derive that $R_t = 0 \leq R_t'$.

Case 2 ($q_t \leq v_t$): We can derive that:
\begin{align}
\nonumber
R_t &= v_t - p_t \mathbf{1}\left\{p_t \leq v_t\right\}\\
\nonumber
    &= v_t - \max\left(q_t, p_t'\right) \mathbf{1}\left\{\max\left(q_t, p_t'\right) \leq v_t\right\}\\
    &= v_t - \max\left(q_t, p_t'\right) \mathbf{1}\left\{p_t' \leq v_t\right\} \label{eq:lem1:1}\\
    &\leq v_t - p_t'\mathbf{1}\left\{p_t' \leq v_t\right\} \label{eq:lem1:2}\\
\nonumber
    &= R_t',
\end{align}
where Equation~(\ref{eq:lem1:1}) follows from that under the antecedent $q_t \leq v_t$, the conditional statement $\{\max(q_t, p_t') \leq v_t \Leftrightarrow q_t \leq v_t$ and $p_t' \leq v_t\}$ can be simplified to $p_t' \leq v_t$. Besides, the inequality in Equation~(\ref{eq:lem1:2}) follows from the maximum function, and takes equal sign when $q_t \leq p_t'$.

Jointly considering the above two cases, we complete the proof.
\end{proof}


According to Lemma~\ref{lem:reserve:price:regret}, for the worst-case regret analysis of Algorithm~\ref{alg:ellipsoid:pricing:noise}, we can use Algorithm~\ref{alg:ellipsoid:pricing:noise}* as its stepping stone. Besides, given the worst-case cumulative regret is obtained under the condition that all the exploratory prices are rejected, we shall focus on the rejection branch of Algorithm~\ref{alg:ellipsoid:pricing:noise} (Lines 14--19). Specifically, the position parameter $\alpha_t$ will utilize the formula in Line 15. Of course, by symmetry, the analysis of the acceptance branch can be similarly derived. Moreover, our key strategy of the following regret analysis is to identify an upper bound on the number of the rounds where the data broker posts the exploratory prices in Algorithm~\ref{alg:ellipsoid:pricing:noise}*, and further show this upper bound still holds in Algorithm~\ref{alg:ellipsoid:pricing:noise}. This converts to tracing the evolution of the ellipsoid's volume, particularly its upper bound and lower bound.

We first introduce a well-studied lemma to construct the upper bound on the ellipsoid's final volume, which is the basis for proving the convergence of the conventional ellipsoid method~\cite{book:ellipsoid:depth:cut} in the field of convex optimization.
\begin{lemma}\label{lem:ellipsoid:volume:ratio:alpha}
Let $\mathcal{E}_t$ denote an ellipsoid, and let $\mathcal{E}_{t+1}$ denote the L\"{o}wner-John ellipsoid obtained after a cut over $\mathcal{E}_t$ with the position parameter $\alpha_t$. If $\alpha_t \in [-1/n, 0]$, we have:
\begin{align}
\frac{V\left(\mathcal{E}_{t+1}\right)}{V\left(\mathcal{E}_{t}\right)} \leq \exp\left(-\frac{\left(1 + n\alpha_t\right)^2}{5n}\right).
\end{align}
\end{lemma}

We next build the lower bound on the ellipsoid's final volume, through monitoring how the smallest eigenvalue of its shape matrix changes. Before that, we introduce a cornerstone, a lemma from~\cite{jour:siam73:golub,book:wilkinson1965algebraic}, to help bound the smallest eigenvalue, if a matrix is modified by another matrix of rank one, \eg, a vector-vector multiplication. In fact, a matrix modified by another rank-one matrix is the update pattern of the ellipsoid's shape matrix in Algorithm~\ref{alg:ellipsoid:pricing:noise} (Lines 17 and 23). Besides, the building block of this lemma is the characteristic polynomial of a square matrix in linear algebra.
\begin{lemma}\label{lem:smallest:eigenvalue}
Let $\mathbf{A}_{n\times n}$ be a symmetric matrix, and let $\gamma_n(\mathbf{A})$ denote its smallest eigenvalue. Let $\mathbf{D} = \mathbf{A} - \beta\mathbf{b}\mathbf{b}^T$ for $\beta > 0$ and $\mathbf{b}\in \mathbb{R}^n$, let $\gamma_n(\mathbf{D})$ denote its smallest eigenvalue, and let $\varphi_{\mathbf{D}}(z)$ denote its characteristic polynomial with the variable $z \in \mathbb{R}$. In particular, when $z$ is not one of the eigenvalues of $\mathbf{A}$, \oie, $z \neq \gamma_i(\mathbf{A}) \forall i \in [n]$, $\varphi_{\mathbf{D}}(z)$ can be expressed as:
\begin{align*}
\varphi_\mathbf{D}\left(z\right) = \det\left(\mathbf{D} - z\mathbf{I}_{n\times n}\right) = \prod_{j\in [n]}\left(\gamma_j\left(\mathbf{A}\right) - z\right) \psi_{\mathbf{D}}\left(z\right),
\end{align*}
where
\begin{align}
\psi_{\mathbf{D}}\left(z\right) = 1 - \beta\sum_{i\in\left[n\right]}\frac{{b_i}^2}{\gamma_i\left(\mathbf{A}\right) - z}.
\end{align}
Then, we have:
\begin{align}
 \forall z < \gamma_n\left(\mathbf{A}\right), \varphi_{\mathbf{D}}\left(z\right) > 0 \Rightarrow \gamma_n\left(\mathbf{D}\right) \geq z.
\end{align}
\end{lemma}

Based on Lemma~\ref{lem:smallest:eigenvalue}, we further give the following two lemmas to construct a lower bound on the smallest eigenvalue of the final ellipsoid's shape matrix, and thus obtain a lower bound on its volume. In particular, Lemma~\ref{lem:smallest:eigenvalue:nondecresing:alpha} shows that if the smallest eigenvalue is below some threshold, namely $\tau \epsilon^2$, it can no longer decrease. Lemma~\ref{lem:smallest:eigenvalue:change:alpha} reveals that in each round, the smallest eigenvalue cannot decrease sharply, to its $\frac{n^2(1-\alpha_t)^2}{(n + 1)^2}$ at most. Therefore, the smallest eigenvalue is bounded below by $\tau\epsilon^2 \frac{n^2(1-\alpha_t)^2}{(n + 1)^2}$. In terms of geometry, these two lemmas follow from that the difference $\bar{p}_t - \ubar{p}_t$ monitors the width of the ellipsoid along the direction given by the feature vector $\mathbf{x}_t$, and if it is below the threshold $\epsilon$, the data broker will post the conservative price rather than the exploratory price to avoid shortening the width along this direction. Hence, the smallest eigenvalue, having a correspondence with the width of the ellipsoid's narrowest axis, cannot become too small.

\begin{lemma}\label{lem:smallest:eigenvalue:nondecresing:alpha}
In Algorithm~\ref{alg:ellipsoid:pricing:noise}* ($\epsilon \geq 4n\delta$), there exists $\tau \in \mathbb{R}$, such that $\gamma_n(\mathbf{A}_t) \leq  \tau\epsilon^2, {\mathbf{x}_t}^T \mathbf{A}_t\mathbf{x}_t > \epsilon^2/4 \Rightarrow \gamma_n(\mathbf{A}_{t+1}) \geq \gamma_n(\mathbf{A}_t)$. Besides, $\tau = \frac{1}{400n^2S^4}$ is a feasible solution.
\end{lemma}
\begin{proof}
We recall that if ${\mathbf{x}_t}^T \mathbf{A}_t\mathbf{x}_t > \epsilon^2/4$, indicating $\bar{p}_t - \ubar{p}_t = 2\sqrt{{\mathbf{x}_t}^T \mathbf{A}_t\mathbf{x}_t} > \epsilon$, the data broker will utilize the effective exploratory price to cut the ellipsoid, where the position parameter is $\alpha_t = -\delta/\sqrt{{\mathbf{x}_t}^T\mathbf{A}_t\mathbf{x}_t}$ in the case of rejection. Since $\epsilon \geq 4n\delta$, implying $0 \leq \delta \leq \frac{\epsilon}{4n}$, we have $\alpha_t \in [-\frac{1}{2n}, 0]$. Besides, from Algorithm~\ref{alg:ellipsoid:pricing:noise}, we can get $\mathbf{A}_{t+1} = \frac{n^2(1 - {\alpha_t}^2)}{n^2 - 1} \mathbf{D}$ for $\mathbf{D} = \mathbf{A}_t - \frac{2(1 + n\alpha_t)}{(n+1)(1 + \alpha_t)} \mathbf{b}_t{\mathbf{b}_t}^T$, where $\mathbf{b}_t = \mathbf{A}_t\mathbf{x}_t/\sqrt{{\mathbf{x}_t}^T\mathbf{A}_t\mathbf{x}_t}$. For convenience, we transform to the base of $\mathbf{A}_t$'s eigenvalues as follows. By using eigendecomposition, we can express
\begin{equation}
\mathbf{A}_t = \Psi\Lambda\Psi^T,
\end{equation}
where $\Psi$ is an orthogonal matrix composed of eigenvectors of $\mathbf{A}_t$, and $\Lambda$ is a diagonal matrix with the eigenvalues of $\mathbf{A}_t$ in its main diagonal. By convention, we sort the diagonal entries of $\Lambda$ in descending order, which implies that the $i$-th diagonal entry of $\Lambda$ is $\gamma_i(\mathbf{A}_t)$ and the $i$-th column of $\Psi$ is the corresponding eigenvector. We then define $\mathbf{A}_{t+1}' = \Psi^T \mathbf{A}_{t+1} \Psi$ and $\mathbf{D}' = \Psi^T \mathbf{D} \Psi$. We can obtain $\mathbf{A}_{t+1}' = \frac{n^2(1- {\alpha_t}^2)}{n^2 - 1}\mathbf{D}'$ and $\mathbf{D}' = \Lambda - \frac{2(1 + n\alpha_t)}{(n+1)(1 + \alpha_t)}\mathbf{b}'{\mathbf{b}'}^T$. Here, $\mathbf{b}' = \Psi \mathbf{b}_t = \Lambda \mathbf{d}/\sqrt{\mathbf{d}^T \Lambda \mathbf{d}}$, where $\mathbf{d} = \Psi^T \mathbf{x}_t$.

Due to the fact that the eigenvalues of any square matrix are invariant under the changing base, we have $\gamma_n(\mathbf{A}_{t+1}) = \gamma_n(\mathbf{A}_{t+1}') = \frac{n^2(1- {\alpha_t}^2)}{n^2 - 1} \gamma_n(\mathbf{D}) = \frac{n^2(1- {\alpha_t}^2)}{n^2 - 1} \gamma_n(\mathbf{D}')$. Now, our task is to prove that under the antecedents $\gamma_n(\mathbf{A}_t) \leq  \tau\epsilon^2$ and ${\mathbf{x}_t}^T \mathbf{A}_t\mathbf{x}_t > \epsilon^2/4$, the consequent $\gamma_n(\mathbf{A}_{t+1}) \geq \gamma_n(\mathbf{A}_t)$, or equivalently $\gamma_n(\mathbf{D}') \geq \frac{n^2-1}{n^2(1 - {\alpha_t}^2)}\gamma_n(\mathbf{A}_t)$, holds. By Lemma~\ref{lem:smallest:eigenvalue} ($\beta$ takes the value $\frac{2(1 + n\alpha_t)}{(n+1)(1 + \alpha_t)}$), we only need to show that $\varphi_{\mathbf{D}'}(\frac{n^2-1}{n^2(1 - {\alpha_t}^2)}\gamma_n(\mathbf{A}_t)) \geq 0$. Given $\alpha_t \in [-\frac{1}{2n}, 0]$, we have $\frac{n^2-1}{n^2(1 - {\alpha_t}^2)}\gamma_n(\mathbf{A}_t) \leq \frac{n^2 - 1}{n^2} \gamma_n(\mathbf{A}_t) < \gamma_n(\mathbf{A}_t)$. Therefore, it suffices to show that: $\psi_{\mathbf{D}'}(\frac{n^2 - 1}{n^2(1 - {\alpha_t}^2)}\gamma_n(\mathbf{A}_t)) \geq 0$. We utilize ${b_i'} = \frac{\gamma_i(\mathbf{A}_t)d_i}{\sqrt{\sum_{j \in [n]} \gamma_j(\mathbf{A}_t) {d_j}^2}}$ to expand $\psi_{\mathbf{D}'}(\frac{n^2 - 1}{n^2(1 - {\alpha_t}^2)}\gamma_n(\mathbf{A}_t))$ as:
\begin{align}
\nonumber
&\psi_{\mathbf{D}'}\left(\frac{n^2 - 1}{n^2\left(1 - {\alpha_t}^2\right)}\gamma_n\left(\mathbf{A}_t\right)\right)\\
\nonumber
=\ &1 - \frac{2\left(1 + n\alpha_t\right)}{\left(n+1\right)\left(1 + \alpha_t\right)}\sum_{i\in\left[n\right]} \frac{{b_i'}^2}{\gamma_i\left(\mathbf{A}_t\right) - \frac{n^2 - 1}{n^2\left(1 - {\alpha_t}^2\right)}\gamma_n\left(\mathbf{A}_t\right)}\\
\nonumber
=\ &1 - \frac{2\left(1 + n\alpha_t\right)}{\left(n+1\right)\left(1 + \alpha_t\right)} \\
   &\quad\ \ \underbrace{\sum_{i \in \left[n\right]} \frac{\gamma_i\left(\mathbf{A}_t\right){d_i}^2}{\sum_{j\in[n]} \gamma_j\left(\mathbf{A}_t\right) {d_j}^2} \frac{1}{1 - \frac{n^2 - 1}{n^2\left(1 - {\alpha_t}^2\right)}\frac{\gamma_n\left(\mathbf{A}_t\right)}{\gamma_i\left(\mathbf{A}_t\right)}}}_{\text{LHS} = \text{LHS1} \left(i: \gamma_i\left(\mathbf{A}_t\right) < \sqrt{\tau}\epsilon^2\right) + \text{LHS2} \left(i: \gamma_i\left(\mathbf{A}_t\right) \geq \sqrt{\tau}\epsilon^2\right)}.\label{eq:LHS}
\end{align}
Depending on whether $\gamma_i(\mathbf{A}_t)$ is less than $\sqrt{\tau}\epsilon^2$ or not, we split LHS in Equation~(\ref{eq:LHS}) into LHS1 and LHS2, and further give upper bounds on these two parts, separately.
\begin{align}
\nonumber
\text{LHS1} &= \sum_{i: \gamma_i(\mathbf{A}_t) < \sqrt{\tau}\epsilon^2} \frac{\gamma_i\left(\mathbf{A}_t\right){d_i}^2}{\sum_{j \in [n]} \gamma_j\left(\mathbf{A}_t\right) {d_j}^2} \frac{1}{1 - \frac{n^2 - 1}{n^2\left(1 - {\alpha_t}^2\right)}\frac{\gamma_n\left(\mathbf{A}_t\right)}{\gamma_i\left(\mathbf{A}_t\right)}},\\
&\leq \sum_{i: \gamma_i\left(\mathbf{A}_t\right) < \sqrt{\tau} \epsilon^2} \frac{\sqrt{\tau} \epsilon^2 {d_i}^2}{\epsilon^2/4}\frac{1}{1 - \frac{n^2 - 1}{n^2\left(1 - {\alpha_t}^2\right)}} \label{eq:lhs1:1}\\
\nonumber
&= 4\sqrt{\tau}\frac{1}{1 - \frac{n^2 - 1}{n^2\left(1 - {\alpha_t}^2\right)}}\sum_{i: \gamma_i\left(\mathbf{A}_t\right) < \sqrt{\tau} \epsilon^2} {d_i}^2\\
\nonumber
&\leq 4\sqrt{\tau}\frac{1}{1 - \frac{n^2 - 1}{n^2\left(1 - {\alpha_t}^2\right)}}\sum_{i\in\left[n\right]} {d_i}^2\\
\nonumber
&= 4\sqrt{\tau}\frac{1}{1 - \frac{n^2 - 1}{n^2\left(1 - {\alpha_t}^2\right)}} {\mathbf{x}_t}^T \Psi^T\Psi\mathbf{x}_t\\
&= 4\sqrt{\tau}\frac{1}{1 - \frac{n^2 - 1}{n^2\left(1 - {\alpha_t}^2\right)}} {\mathbf{x}_t}^T\mathbf{x}_t \label{eq:lhs1:2}\\
&\leq 4\sqrt{\tau} \frac{1}{1 - \frac{n^2 - 1}{n^2\left(1 - {\alpha_t}^2\right)}} S^2.\label{eq:lhs1:3}
\end{align}
Here, the inequality in Equation~(\ref{eq:lhs1:1}) follows from $\gamma_i(\mathbf{A}_t) < \sqrt{\tau}\epsilon^2$ and $\sum_{j \in [n]} \gamma_j(\mathbf{A}_t){d_j}^2 = {\mathbf{x}_t}^T\mathbf{A}_t\mathbf{x}_t > \epsilon^2/4$ in the antecedents, and $\gamma_n(\mathbf{A}_t) \leq \gamma_i(\mathbf{A}_t)$. Besides, Equation~(\ref{eq:lhs1:2}) follows from the orthogonality of $\Psi$. Moreover, the inequality in Equation~(\ref{eq:lhs1:3}) follows from $\|\mathbf{x}_t\| \leq S$.
\begin{align}
\nonumber
\text{LHS2} &= \sum_{i: \gamma_i(\mathbf{A}_t) \geq \sqrt{\tau}\epsilon^2} \frac{\gamma_i\left(\mathbf{A}_t\right){d_i}^2}{\sum_{j \in [n]} \gamma_j\left(\mathbf{A}_t\right) {d_j}^2} \frac{1}{1 - \frac{n^2 - 1}{n^2\left(1 - {\alpha_t}^2\right)}\frac{\gamma_n\left(\mathbf{A}_t\right)}{\gamma_i\left(\mathbf{A}_t\right)}}\\
\nonumber
&\leq \sum_{i: \gamma_i\left(\mathbf{A}_t\right) \geq \sqrt{\tau} \epsilon^2} \frac{\gamma_i\left(\mathbf{A}_t\right){d_i}^2}{\sum_{j \in [n]} \gamma_j\left(\mathbf{A}_t\right) {d_j}^2} \frac{1}{1 - \frac{n^2 - 1}{n^2\left(1 - {\alpha_t}^2\right)}\sqrt{\tau}}\\
\nonumber
&= \frac{1}{1 - \frac{n^2 - 1}{n^2\left(1 - {\alpha_t}^2\right)}\sqrt{\tau}} \sum_{i: \gamma_i\left(\mathbf{A}_t\right) \geq \sqrt{\tau} \epsilon^2} \frac{\gamma_i\left(\mathbf{A}_t\right){d_i}^2}{\sum_{j \in [n]} \gamma_j\left(\mathbf{A}_t\right) {d_j}^2}\\
\nonumber
&\leq \frac{1}{1 - \frac{n^2 - 1}{n^2\left(1 - {\alpha_t}^2\right)}\sqrt{\tau}} \sum_{i \in \left[n\right]} \frac{\gamma_i\left(\mathbf{A}_t\right){d_i}^2}{\sum_{j \in [n]} \gamma_j\left(\mathbf{A}_t\right) {d_j}^2}\\
\nonumber
&= \frac{1}{1 - \frac{n^2 - 1}{n^2\left(1 - {\alpha_t}^2\right)}\sqrt{\tau}},
\end{align}
where the first inequality follows from $\gamma_n\left(\mathbf{A}_t\right) \leq \tau\epsilon^2$ in the antecedent and $\gamma_i\left(\mathbf{A}_t\right) \geq \sqrt{\tau}\epsilon^2$.

By combining two upper bounds on LHS1 and LHS2, we finally derive a lower bound on $\psi_{\mathbf{D}'}(\frac{n^2 - 1}{n^2(1 - {\alpha_t}^2)}\gamma_n(\mathbf{A}_t))$ as the parameter $\tau$ approaches 0:
\begin{align*}
   &\psi_{\mathbf{D}'}\left(\frac{n^2 - 1}{n^2\left(1 - {\alpha_t}^2\right)}\gamma_n\left(\mathbf{A}_t\right)\right)\\
=\ &1 - \frac{2\left(1 + n\alpha_t\right)}{\left(n+1\right)\left(1 + \alpha_t\right)} \left(\text{LHS1} + \text{LHS2}\right)\\
\geq\ &1 - \frac{2\left(1 + n\alpha_t\right)}{\left(n+1\right)\left(1 + \alpha_t\right)}\underbrace{ \left(\frac{4\sqrt{\tau}S^2}{1 - \frac{n^2 - 1}{n^2\left(1 - {\alpha_t}^2\right)}} + \frac{1}{1 - \frac{n^2 - 1}{n^2\left(1 - {\alpha_t}^2\right)}\sqrt{\tau}}\right)}_{= 1,\ \text{as}\ \alpha_t \in \left[-\frac{1}{2n}, 0\right]\ \text{and}\ \lim_{\tau \rightarrow 0}}\\
\geq\ &1 - \frac{2\left(1 + n\alpha_t\right)}{\left(n + 1\right)\left(1 + \alpha_t\right)}\\
=\    &\frac{\left(n - 1\right)\left(1 - \alpha_t\right)}{\left(n + 1\right)\left(1 + \alpha_t\right)}\\
\geq\ &0.
\end{align*}
In addition, we can check $\tau = \frac{1}{400n^2S^4}$ is a feasible value to let the lower bound and thus the original $\psi_{\mathbf{D}'}(\frac{n^2 - 1}{n^2(1 - {\alpha_t}^2)}\gamma_n(\mathbf{A}_t))$ be no less than 0. This completes the proof.
\end{proof}

\begin{lemma}\label{lem:smallest:eigenvalue:change:alpha}
For any round $t$ in Algorithm~\ref{alg:ellipsoid:pricing:noise}* ($\epsilon \geq 4n\delta$), where the data broker posts the exploratory price, we have: $\gamma_{n}(\mathbf{A}_{t+1}) \geq \frac{n^2(1-\alpha_t)^2}{(n + 1)^2}\gamma_n\left(\mathbf{A}_t\right)$.
\end{lemma}
\begin{proof}
We continue to use the notations in the proof of Lemma~\ref{lem:smallest:eigenvalue:nondecresing:alpha}. We need to prove $\gamma_{n}(\mathbf{A}_{t+1}) \geq \frac{n^2(1-\alpha_t)^2}{(n + 1)^2}\gamma_n(\mathbf{A}_t)$. It suffices to prove $\gamma_{n}(\mathbf{A}_{t+1})= \frac{n^2(1- {\alpha_t}^2)}{n^2 - 1} \gamma_n(\mathbf{D}') \geq \frac{n^2(1-\alpha_t)^2}{(n + 1)^2}\gamma_n(\mathbf{A}_t)$, or $\gamma_n(\mathbf{D}') \geq \frac{(n - 1)(1 - \alpha_t)}{(n + 1)(1 + \alpha_t)}\gamma_n(\mathbf{A}_t)$. By Lemma~\ref{lem:smallest:eigenvalue}, it suffices to prove $\varphi_{\mathbf{D}'}\left(\frac{(n - 1)(1 - \alpha_t)}{(n + 1)(1 + \alpha_t)}\gamma_n(\mathbf{A}_t)\right) \geq 0$. Given $\alpha_t \in [-\frac{1}{2n}, 0]$, we have: $\frac{(n - 1)(1 - \alpha_t)}{(n + 1)(1 + \alpha_t)}\gamma_n(\mathbf{A}_t) \leq \frac{2n^2 - n - 1}{2n^2 +n - 1} \gamma_n(\mathbf{A}_t) < \gamma_n(\mathbf{A}_t)$. Thus, we only need to prove $\psi_{\mathbf{D}'}\left(\frac{(n - 1)(1 - \alpha_t)}{(n + 1)(1 + \alpha_t)}\gamma_n(\mathbf{A}_t)\right) \geq 0,$ which holds as follows.
\begin{align}
\nonumber
&\psi_{\mathbf{D}'}\left(\frac{\left(n - 1\right)\left(1 - \alpha_t\right)}{\left(n + 1\right)\left(1 + \alpha_t\right)}\gamma_n\left(\mathbf{A}_t\right)\right)\\
\nonumber
=\ &1 - \frac{2\left(1 + n\alpha_t\right)}{\left(n+1\right)\left(1 + \alpha_t\right)} \\
\nonumber
   &\quad\ \ \sum_{i \in \left[n\right]} \frac{\gamma_i\left(\mathbf{A}_t\right){d_i}^2}{\sum_{j \in [n]} \gamma_j\left(\mathbf{A}_t\right) {d_j}^2} \frac{1}{1 - \frac{(n - 1)(1 - \alpha_t)}{(n + 1)(1 + \alpha_t)}\frac{\gamma_n\left(\mathbf{A}_t\right)}{\gamma_i\left(\mathbf{A}_t\right)}}\\
\geq\ &1 - \frac{2\left(1 + n\alpha_t\right)}{\left(n+1\right)\left(1 + \alpha_t\right)} \max_{i \in \left[n\right]} \frac{1}{1 - \frac{\left(n - 1\right)\left(1 - \alpha_t\right)}{\left(n + 1\right)\left(1 + \alpha_t\right)}\frac{\gamma_n\left(\mathbf{A}_t\right)}{\gamma_i\left(\mathbf{A}_t\right)}} \label{eq:lem9:1}\\
=\ &1 - \frac{2\left(1 + n\alpha_t\right)}{\left(n+1\right)\left(1 + \alpha_t\right)} \frac{1}{1 - \frac{\left(n - 1\right)\left(1 - \alpha_t\right)}{\left(n + 1\right)\left(1 + \alpha_t\right)}\frac{\gamma_n\left(\mathbf{A}_t\right)}{\gamma_n\left(\mathbf{A}_t\right)}}\label{eq:lem9:2}\\
\nonumber
=\ &0.
\end{align}
Here, the inequality in Equation~(\ref{eq:lem9:1}) follows from that the terms $\frac{\gamma_i\left(\mathbf{A}_t\right){d_i}^2}{\sum_{j \in [n]} \gamma_j\left(\mathbf{A}_t\right) {d_j}^2}$'s stand for the coefficients of a convex combination of elements, which can be bounded above by the maximum element. Besides, Equation~(\ref{eq:lem9:2}) follows from that $\gamma_n(\mathbf{A}_t)$ is the smallest eigenvalue.
\end{proof}

By combining the upper and lower bounds on the final ellipsoid's volume, we can derive an upper bound on the number of rounds where the exploratory prices are posted, hereinafter denoted as $T_e$, in Algorithm~\ref{alg:ellipsoid:pricing:noise}*.

\begin{lemma}\label{lem:exploratory:number:round:alpha}
Algorithm~\ref{alg:ellipsoid:pricing:noise}* ($\epsilon \geq 4n\delta$) chooses the exploratory prices in at most $20n^2\log(20RS^2(n + 1)/\epsilon)$ rounds.
\end{lemma}
\begin{proof}
In Algorithm~\ref{alg:ellipsoid:pricing:noise}*, each effective exploratory price incurs a cut with the position parameter $\alpha_t \in [-\frac{1}{2n}, 0]$ over the ellipsoid, while each conservative price keeps the ellipsoid unchanged. We next give an upper bound and a lower bound on the ratio between the volumes of the final and the initial ellipsoids, namely $V(\mathcal{E}_{T+1})/V(\mathcal{E}_{1})$.

For the upper bound, by Lemma~\ref{lem:ellipsoid:volume:ratio:alpha}, after each round $t$ where the exploratory price is posted, we have:
\begin{align}
\frac{V\left(\mathcal{E}_{t+1}\right)}{V\left(\mathcal{E}_{t}\right)} \leq \exp\left(-\frac{\left(1 + n\alpha_t\right)^2}{5n}\right) \leq \exp\left(-\frac{1}{20n}\right),
\end{align}
where the last inequality takes equal sign at $\alpha_t = -\frac{1}{2n}$. Hence, after $T$ rounds, we can derive that:
\begin{align}
\frac{V\left(\mathcal{E}_{T+1}\right)}{V\left(\mathcal{E}_1\right)} \leq \exp\left(- \frac{T_e}{20n}\right).
\end{align}

Regarding the lower bound, by Lemmas~\ref{lem:smallest:eigenvalue:nondecresing:alpha} and \ref{lem:smallest:eigenvalue:change:alpha}, we have:
\begin{align}
\gamma_n\left(\mathbf{A}_{T+1}\right) \geq \tau\epsilon^2\frac{n^2(1-\alpha_t)^2}{(n + 1)^2} \geq \tau\epsilon^2\frac{n^2}{(n + 1)^2},\label{eq:ellipsoid:smallest:eigenvalue}
\end{align}
where the last inequality takes equal sign at $\alpha_t = 0$. Correspondingly, we can obtain a lower bound on $V(\mathcal{E}_{T+1})$:
\begin{align}
V(\mathcal{E}_{T+1}) = V_n \sqrt{\prod_{i}^{n} \gamma_i\left(\mathbf{A}_{T+1}\right)} \geq V_n \left(\sqrt{\tau} \epsilon \frac{n}{n+1}\right)^{n}.
\end{align}
Besides, the initial ellipsoid $\mathcal{E}_1$ is an $n$-dimensional ball with radius $R$, whose volume is $V_n R^n$. Therefore, we can derive a lower bound on $V(\mathcal{E}_{T+1})/V(\mathcal{E}_1)$:
\begin{align}
\frac{V(\mathcal{E}_{T+1})}{V(\mathcal{E}_1)} \geq \left(\frac{\sqrt{\tau} \epsilon}{R} \frac{n}{n+1}\right)^n.
\end{align}

By using the inequality that the upper bound is no less than the lower bound, we can get an upper bound on $T_e$:
\begin{equation}
\begin{aligned}
            &\left(\frac{\sqrt{\tau} \epsilon}{R} \frac{n}{n+1}\right)^n \leq \exp\left(- \frac{T_e}{20n}\right)\\
\Rightarrow\quad &T_e\ \leq\ 20n^2\log\left(\frac{R}{\sqrt{\tau} \epsilon} \frac{n + 1}{n}\right).
\end{aligned}
\end{equation}
We note that this upper bound holds for any feasible $\tau$. When $\tau$ takes the value $\frac{1}{400n^2S^4}$ from Lemma~\ref{lem:smallest:eigenvalue:nondecresing:alpha}, we have:
\begin{align}
T_e \leq 20n^2\log\left(20 RS^2\left(n+1\right)/\epsilon\right).
\end{align}
This completes the proof.
\end{proof}

We finally analyze the worst-case regret of Algorithm~\ref{alg:ellipsoid:pricing:noise}. Before that, we restate Lemma~\ref{lem:exploratory:number:round:alpha} for Algorithm~\ref{alg:ellipsoid:pricing:noise}, by analyzing the impact of the reserve price constraint on $T_e$.

\begin{lemma}\label{lem:original:exploratory:number:round}
Algorithm~\ref{alg:ellipsoid:pricing:noise} ($\epsilon \geq 4n\delta$) chooses the exploratory prices in at most $20n^2\log(20RS^2(n + 1)/\epsilon)$ rounds.
\end{lemma}
\begin{proof}
We recall that if the reserve price $q_t$ is introduced in round $t$, the exploratory posted price is $p_t = \max(q_t, \frac{\ubar{p}_t + \bar{p}_t}{2})$, the effective exploratory price is $p_t + \delta$ in the rejection case, and its position parameter can be computed via $\alpha_t =  (\frac{\ubar{p}_t + \bar{p}_t}{2} - (p_t + \delta))/\sqrt{{\mathbf{x}_t}^T \mathbf{A}_t \mathbf{x}_t}$ (Algorithm~\ref{alg:ellipsoid:pricing:noise}, Line 15). We now prove Lemma~\ref{lem:original:exploratory:number:round} in two complementary cases as follows.

Case 1 $(\frac{\ubar{p}_t + \bar{p}_t}{2} \geq q_t)$: The posted price is the middle price, \ie, $p_t = \frac{\ubar{p}_t + \bar{p}_t}{2}$. Algorithm~\ref{alg:ellipsoid:pricing:noise} degenerates to Algorithm~\ref{alg:ellipsoid:pricing:noise}*, and Lemma~\ref{lem:original:exploratory:number:round} holds from Lemma~\ref{lem:exploratory:number:round:alpha}.

Case 2 $(q_t > \frac{\ubar{p}_t + \bar{p}_t}{2})$: The posted price is the reserve price, \ie, $p_t = q_t$. Given the reserve price is rejected, we can draw that the reserve price is higher than the market value, \ie, $p_t = q_t > v_t$, which further implies that $R_t = 0$ from Equation~(\ref{eq:regret:formula:rp}). Suppose the data broker does not utilize the reserve price to refine the ellipsoid in this round. The analysis of Algorithm~\ref{alg:ellipsoid:pricing:noise} can be reduced to analyzing Algorithm~\ref{alg:ellipsoid:pricing:noise}* with the number of total rounds $T - 1$ plus one dummy round inserted in the $t$-th round. Considering Lemma~\ref{lem:exploratory:number:round:alpha} does not rely on the number of total rounds, $T_e \leq 20n^2\log(20RS^2(n + 1)/\epsilon)$ still holds in Algorithm~\ref{alg:ellipsoid:pricing:noise}. However, in Algorithm~\ref{alg:ellipsoid:pricing:noise} (Lines 14--19), the data broker needs to cut the ellipsoid using the effective exploratory price, \ie, $q_t + \delta$ here. We thus need to analyze the impact of such a cut on $T_e$. In particular, following the guidelines in proving Lemma~\ref{lem:exploratory:number:round:alpha}, to prove Lemma~\ref{lem:original:exploratory:number:round}, it suffices to prove that this cut cannot increase the upper bound on the final volume of the ellipsoid, and meanwhile, cannot decrease the lower bound. First, the effective exploratory price imposes a cut over the ellipsoid, and thus cannot increase the final volume together with the upper bound on the final volume. Second, in Equation~(\ref{eq:ellipsoid:smallest:eigenvalue}), the lower bound on the smallest eigenvalue of the final ellipsoid's shape matrix is obtained at $\alpha_t = 0$. Besides, we can check that a negative $\alpha_t$ increases the lower bound. Therefore, the effective exploratory price $q_t + \delta$, holding a negative $\alpha_t =  (\frac{\ubar{p}_t + \bar{p}_t}{2} - (q_t + \delta))/\sqrt{{\mathbf{x}_t}^T \mathbf{A}_t \mathbf{x}_t} < -  \delta/\sqrt{{\mathbf{x}_t}^T \mathbf{A}_t \mathbf{x}_t} < 0$, cannot decrease the lower bound on the final volume of the ellipsoid.

By summarizing two cases, we complete the proof.
\end{proof}

\begin{theorem}\label{theorem:alg2}
If $\delta = O(n/T)$, then the worst-case regret of Algorithm~\ref{alg:ellipsoid:pricing:noise} is $O(\max(n^2\log(T/n), n^3\log(T/n)/T))$.\footnote{The regret analysis in the one-dimensional case is deferred to Theorem~\ref{them:one:dimension} in Appendix. Besides, our regret analyses in this work focus on the terms, involving the dimension of the feature vector $n$ and the number of total rounds $T$, and ignore other constant terms.}
\end{theorem}
\begin{proof}
First, as we illustrated in the paragraph below Equation~(\ref{eq:noise:bound}), in each round $t$, the absolute value of the random variable $\delta_t$ has probability at most $1/T$ outside $\delta$. Thus, the cumulative regret incurred by removing the weight vector $\theta^*$ from the knowledge set is at most $\max_{\mathbf{x}_t, {\theta^*}} {\mathbf{x}_t}^T{\theta^*} T / T = RS$.

Second, we analyze the cumulative regret due to the posted prices. In round $t$, the regret incurred by posting the exploratory (\resp, conservative) price can be bounded above by $\bar{p}_t + \delta$ (\resp, $(\bar{p}_t + \delta) - (\ubar{p}_t - \delta)$), which can be further bounded above by $RS + \delta$ (\resp, $\epsilon + 2\delta$). Thus, the cumulative regret is no more than $T_e(RS + \delta) + (T - T_e)(\epsilon + 2\delta)$. When $\delta = O(n/T)$, $T_e$ takes its upper bound $20n^2\log(20RS^2$ $(n + 1)/\epsilon)$ from Lemma~\ref{lem:original:exploratory:number:round}, and $\epsilon$ is set to $\max(n^2/T, 4n\delta) = O(n^2/T)$, the worst-case regret incurred by the posted prices is $O(\max(n^2\log(T/n), n^3\log(T/n)/T))$.

By adding the above two parts, the worst-case regret of Algorithm~\ref{alg:ellipsoid:pricing:noise} is $O(\max(n^2\log(T/n), n^3\log(T/n)/T))$. This completes the proof.
\end{proof}


\section{Extensions}\label{sec:extensions}

In this section, we extend our proposed pricing framework under the fundamental linear model to some other non-linear models. We also discuss how to support several other similar application scenarios. 

\subsection{Non-linear Market Value Models}
We mainly investigate four kinds of non-linear models commonly used in measuring market values. The first two are the log-log and log-linear models in hedonic pricing from real estate and property studies~\cite{jour:hedonic:milon1984,jour:hedonic:malpezzi2002,jour:hedonic:sirmans2005,proc:kdd18:airbnb}, which can be formalized as $\log v_t = \sum_{i \in [n]} \log(x_{t, i}) {\theta_i^*}$ and $\log v_t = {{\mathbf{x}_t}^T}\theta^*$, respectively. Here, $x_{t,i}$ and $\theta_i^*$ denote the $i$-th elements of the feature vector $\mathbf{x}_t$ and the weight vector $\theta^*$, respectively. The other two models are the logistic model~\cite{proc:www07:richardson:lr,proc:www08:chakrabarti:lr,proc:kdd13:mcmahan} and the kernelized model~\cite{proc:nips14:amin} in online advertising, which can be formalized as $v_t = 1/(1 + \exp({\mathbf{x}_t}^T\theta^*))$ and $v_t = \sum_{k = 1}^{t - 1} K(\mathbf{x}_t, \mathbf{x}_k) {\theta_k^*}$, respectively. Here, $K(\cdot, \cdot)$ is a Mercer kernel operator.


We can further observe that the above four non-linear models can be unified to a general class of non-linear models:
\begin{align}
v_t = g\left(\phi\left(\mathbf{x}_t\right)^T \theta^*\right),
\end{align}
where $g: \mathbb{R} \mapsto \mathbb{R}$ is a non-decreasing and continuous function, \eg, in the two hedonic pricing models, $g$ is the natural exponential function; in the logistic model, $g$ is the logistic sigmoid function; in the kernelized model, $g$ is the identity function. Besides, $\phi: \mathbb{R}^n \mapsto \mathbb{R}^m$ represents a feature mapping of the original feature vector $\mathbf{x}_t$, which intends to capture non-linear correlations/dependencies among the different features of $\mathbf{x}_t$ and the different feature vectors within $t$ rounds. For example, in the log-log model, $\phi$ denotes applying the natural logarithm function to each element of $\mathbf{x}_t$; in the kernelized model, $m = t - 1$ and $\phi$ stands for the kernel function $K$; in the other two models, $\phi$ denotes the identity map. Furthermore, it is worth noting that both $g$ and $\phi$ are public knowledge, and only the weight vector $\theta^*$ is unknown. Therefore, by regarding the domain of $\theta^*$ as the knowledge set to be refined, our proposed pricing mechanism under the linear model can still apply to the above class of non-linear models. Specifically, $\phi(\mathbf{x}_t)$ now functions as the new feature vector, and the threshold $\epsilon$ is used to control $\bar{p}_t - \ubar{p}_t$, \ie, the difference between the maximum and minimum possible values of $\phi(\mathbf{x}_t)^T\theta$, where $\theta$ belongs to the data broker's knowledge set. In addition, the data broker will post the price $g(p_t)$ rather than the original $p_t$. We finally analyze the worst-case regret of Algorithm~\ref{alg:ellipsoid:pricing:noise} when adapted to the above class of non-linear models.

\begin{theorem}
Let $g: \mathbb{R} \mapsto \mathbb{R}$ denote a non-decreasing and continuous function with Lipschitz constant $L$. Let $\phi: \mathbb{R}^n \mapsto \mathbb{R}^m$ denote a mapping of the feature vector $\mathbf{x}_t$ such that $\|\phi(\mathbf{x}_t)\| \leq U$ for some $U \in \mathbb{R}$, and let $\theta^*$ denote the weight vector over $\phi(\mathbf{x}_t)$ such that $\|\theta^*\| \leq R$. If the market value takes the form $v_t = g(\phi(\mathbf{x}_t)^T \theta^*)$, and the uncertainty parameter $\delta = O(n/T)$, the worst-case regret of the adapted Algorithm~\ref{alg:ellipsoid:pricing:noise} is $O(\max(n^2\log(T/n), n^3\log(T/n)/T))$.
\end{theorem}
\begin{proof}
First, the cumulative regret incurred by removing the weight vector $\theta^*$ from the knowledge set is at most $g(RU)T/T = g(RU)$.

Second, we analyze the cumulative regret caused by the posted prices. In round $t$, the regret incurred by posting the exploratory (\resp, conservative) price is no more than $g(\bar{p}_t + \delta)$ (\resp, $g(\bar{p}_t + \delta) - g(\ubar{p}_t - \delta)$), which can be further bounded above by $g(RU + \delta)$ (\resp, $L ((\bar{p}_t + \delta) - (\ubar{p}_t - \delta)) \leq L(\epsilon + 2\delta)$). Therefore, the cumulative regret is no more than $T_e g(RU + \delta) + (T - T_e)L(\epsilon + 2\delta)$.
When  $\delta = O(n/T)$, $T_e$ takes its upper bound $20n^2\log(20RU^2(n+1)/\epsilon)$ (here, $U$ replaces $S$ in Lemma~\ref{lem:original:exploratory:number:round} for the new feature vector $\phi(\mathbf{x}_t)$), and $\epsilon = O(n^2/T)$, the cumulative regret incurred by the posted prices is $O(\max(n^2\log(T/n), n^3\log(T/n)/T))$.

By adding the above two parts, we complete the proof.
\end{proof}

\subsection{Supporting Other Application Scenarios}

We first outline the characteristics of the pricing problem in online personal data markets. We then point out some other similar application scenarios in practice, and further illustrate how to support them with our proposed pricing mechanism under different market value models.

In data markets, the data broker is the seller, and each data consumer is a buyer. Besides, the sequential queries, as the products to be sold, have three atypical characteristics: (1) {\em Customization}: The queries, requested by different data consumers, are highly differentiated; (2) {\em Existence of reserve price}: The total privacy compensation, allocated to the underlying data owners, serves as the reserve price of a query; (3) {\em Timeliness}: If no deal occurs in a round, the query will vanish, generating regret at the data broker.


Several other products in practice share one or more characteristics listed above, which implies that our proposed pricing mechanism for personal data markets can be extended to these scenarios. One example is hospitality service on booking platforms, \eg, Airbnb, Wimdu, and Workaway. A tourist can raise some requirements on her desirable accommodation, such as location, the numbers of bedrooms and bathrooms, amenities, reviews, historical occupancy rate, and so on. Meanwhile, the host of the house can set a minimum/reserve price for the accommodation. If the house is not rented out at a certain date, it may cause regret at both the host and the booking platform. We can observe that the host, the booking platform, and the tourist, play similar roles to the data owner, the data broker, and the data consumer in data markets, respectively. Besides, the market value of the accommodation can be well interpreted by the linear or log-linear model~\cite{proc:kdd18:airbnb}. Yet, another example is online advertising on web publishers. Here, we consider a novel scenario, where the impressions are traded through posting prices rather than the ad auctions already adopted by Internet giants, \eg, Google, Microsoft, and Facebook. In particular, an advertiser can customize her need of an impression, \eg, position and target audience. If the impression is not sold within a given time frame, it will generate regret at the web publisher. We note that the web visitors who generate impressions, the web publisher, and the advertiser, play similar roles to the data owners, the data broker, and the data consumer in data markets, respectively. Besides, the market value of an impression is normally measured by its click-through rate (CTR), which can be effectively captured by the logistic~\cite{proc:www07:richardson:lr,proc:www08:chakrabarti:lr,proc:kdd13:mcmahan} or kernelized model~\cite{proc:nips14:amin}. Last but not least example is loan application from banks and other financial institutions, \eg, JPMorgan Chase Bank, Wells Fargo, and Lending Club. A borrower intends to request a certain amount of loan. Based on the situations of the borrower, including the credit score, employment status, and property state, a financial institution determines whether to approve the application and at which interest rate. Of course, the borrower can also accept or reject the loan offer. The case of rejection may incur regret at the financial institution. Here, the financial institution and the borrower play similar roles to the data broker and the data consumer in data markets, respectively. In addition, the interest rate of a loan can be captured by the linear or log-log model~\cite{jor:jar98:loan:pricing,book:2005:loan}.



In conclusion, our proposed pricing mechanism is not just limited to online personal data trading, and can also support other application scenarios with common characteristics.



\section{Evaluation Results}\label{sec:evaluations}

\begin{figure*}[t]
\centering
\subfigure[$n = 1$]{\label{fig:query:n1}
\includegraphics[width=0.66\columnwidth]{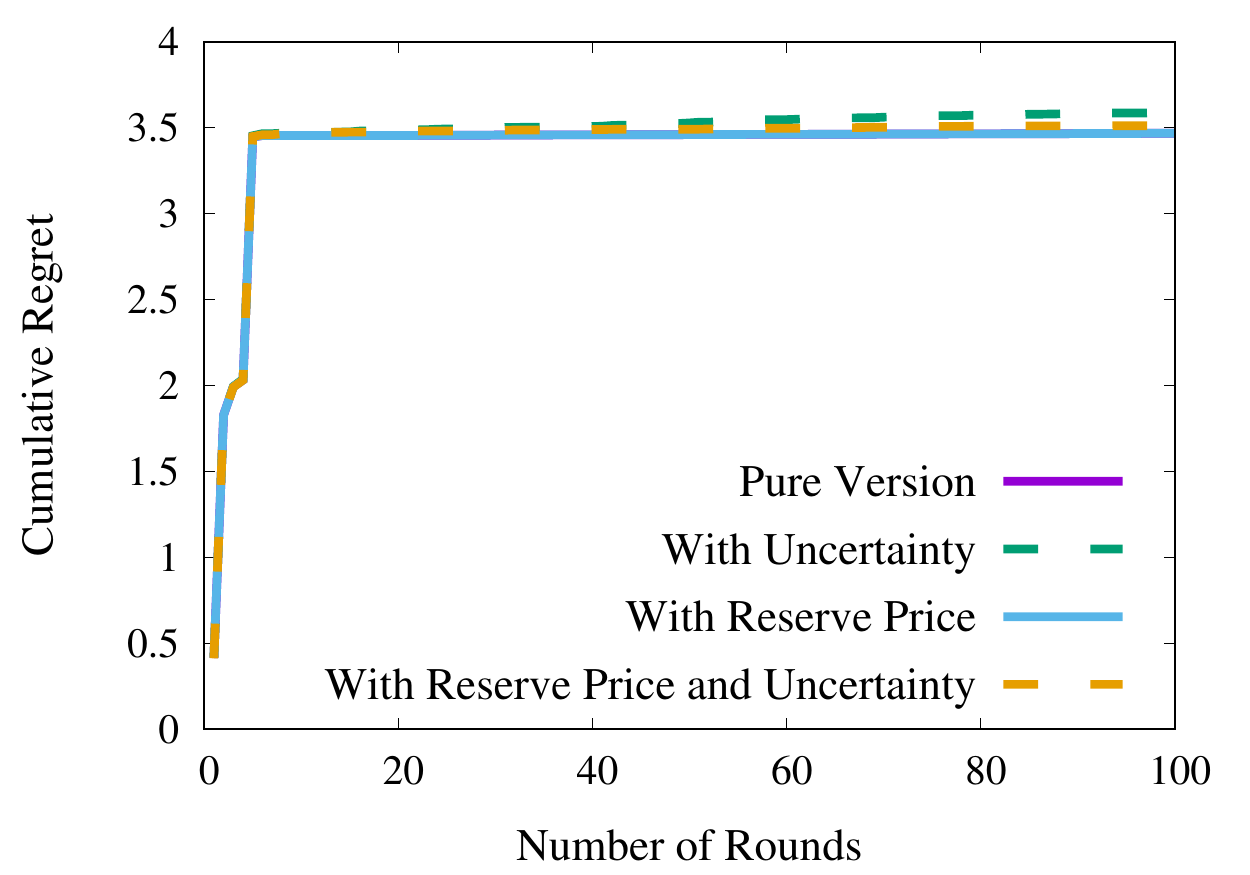}}
\subfigure[$n = 20$]{\label{fig:query:n20}
\includegraphics[width=0.66\columnwidth]{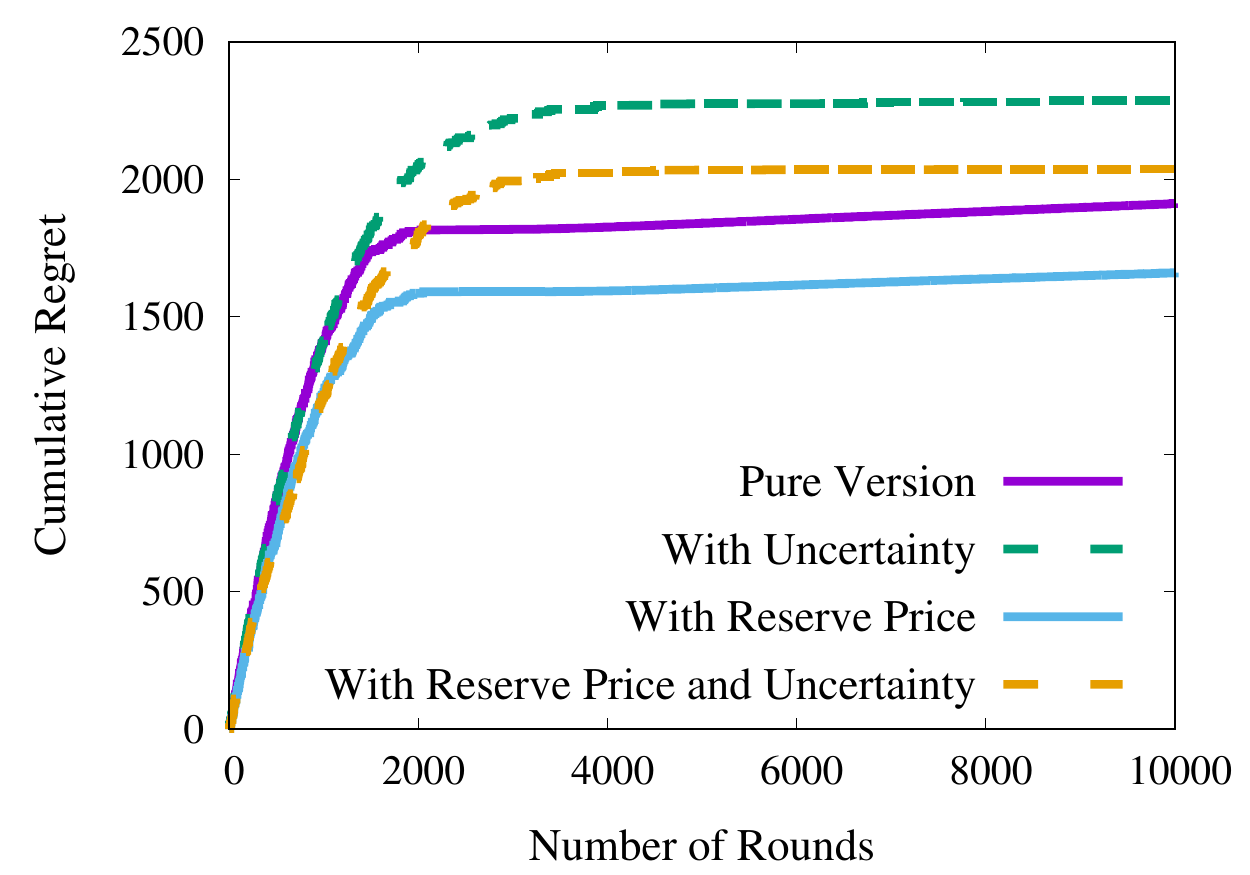}}
\subfigure[$n = 40$]{\label{fig:query:n40}
\includegraphics[width=0.66\columnwidth]{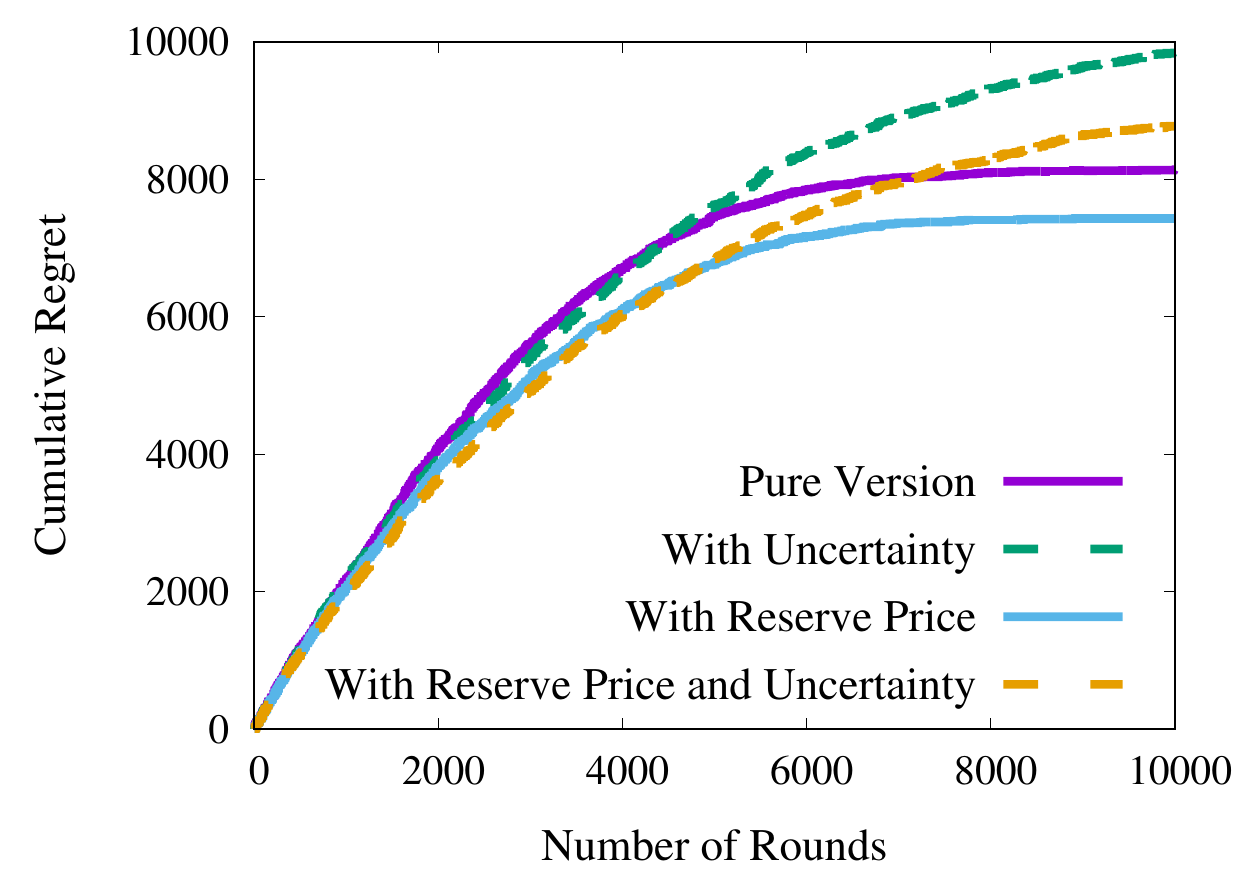}}
\subfigure[$n = 60$]{\label{fig:query:n60}
\includegraphics[width=0.66\columnwidth]{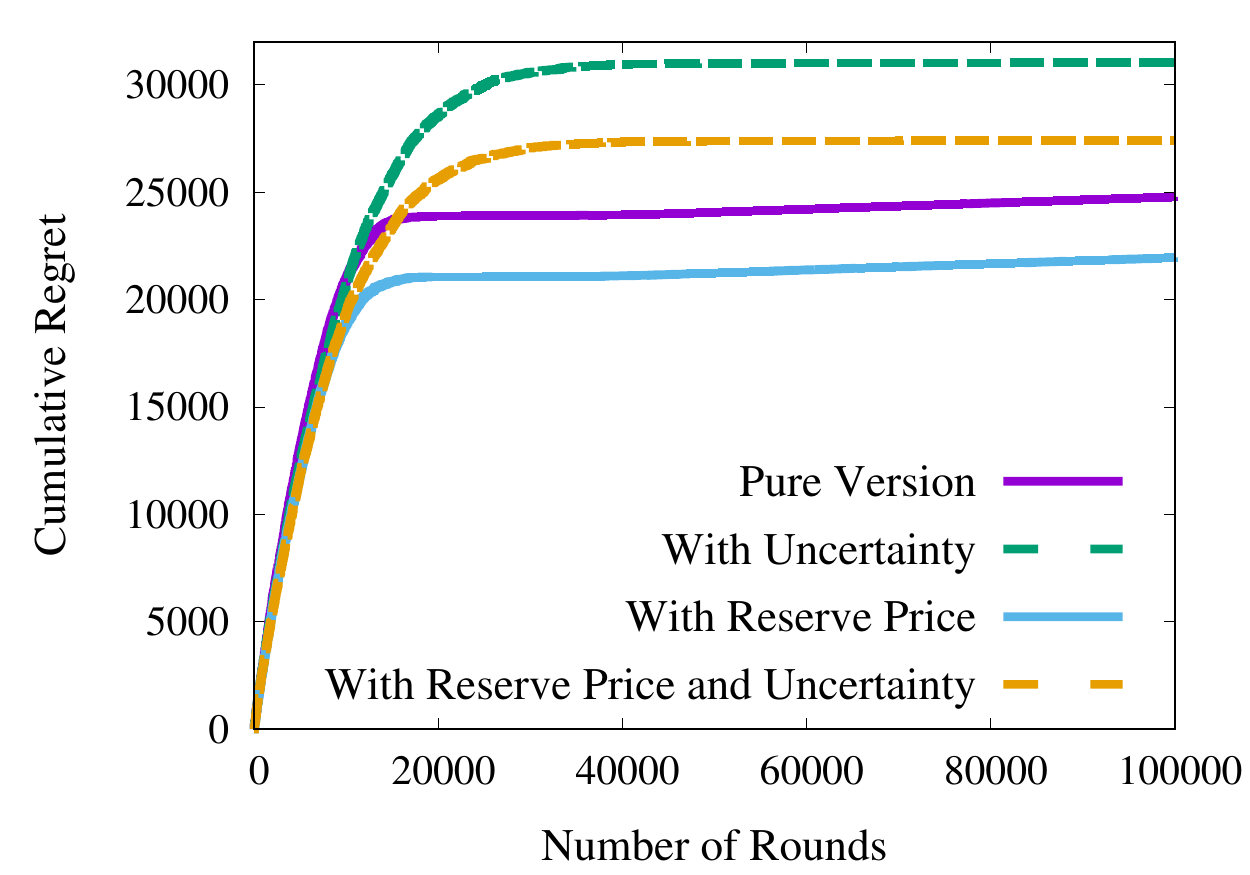}}
\subfigure[$n = 80$]{\label{fig:query:n80}
\includegraphics[width=0.66\columnwidth]{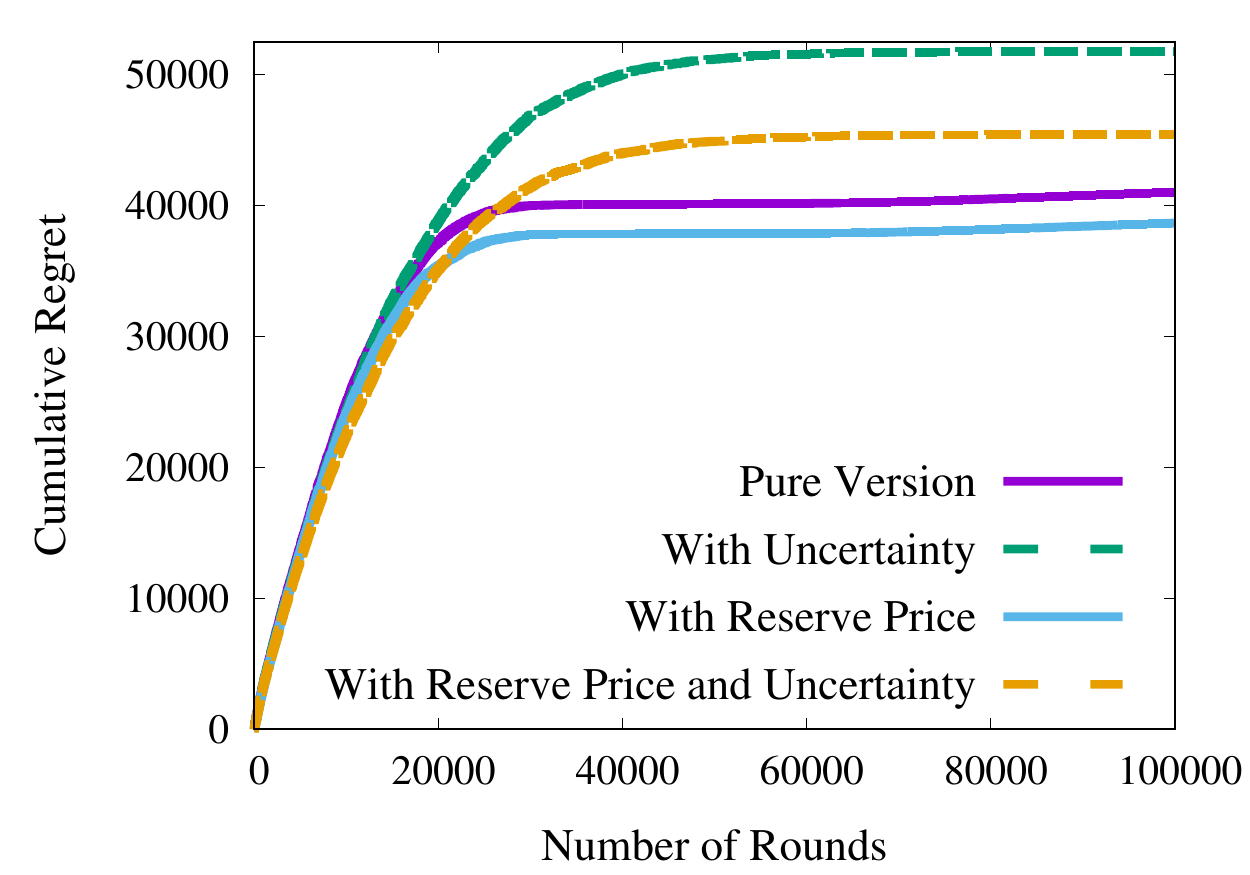}}
\subfigure[$n = 100$]{\label{fig:query:n100}
\includegraphics[width=0.66\columnwidth]{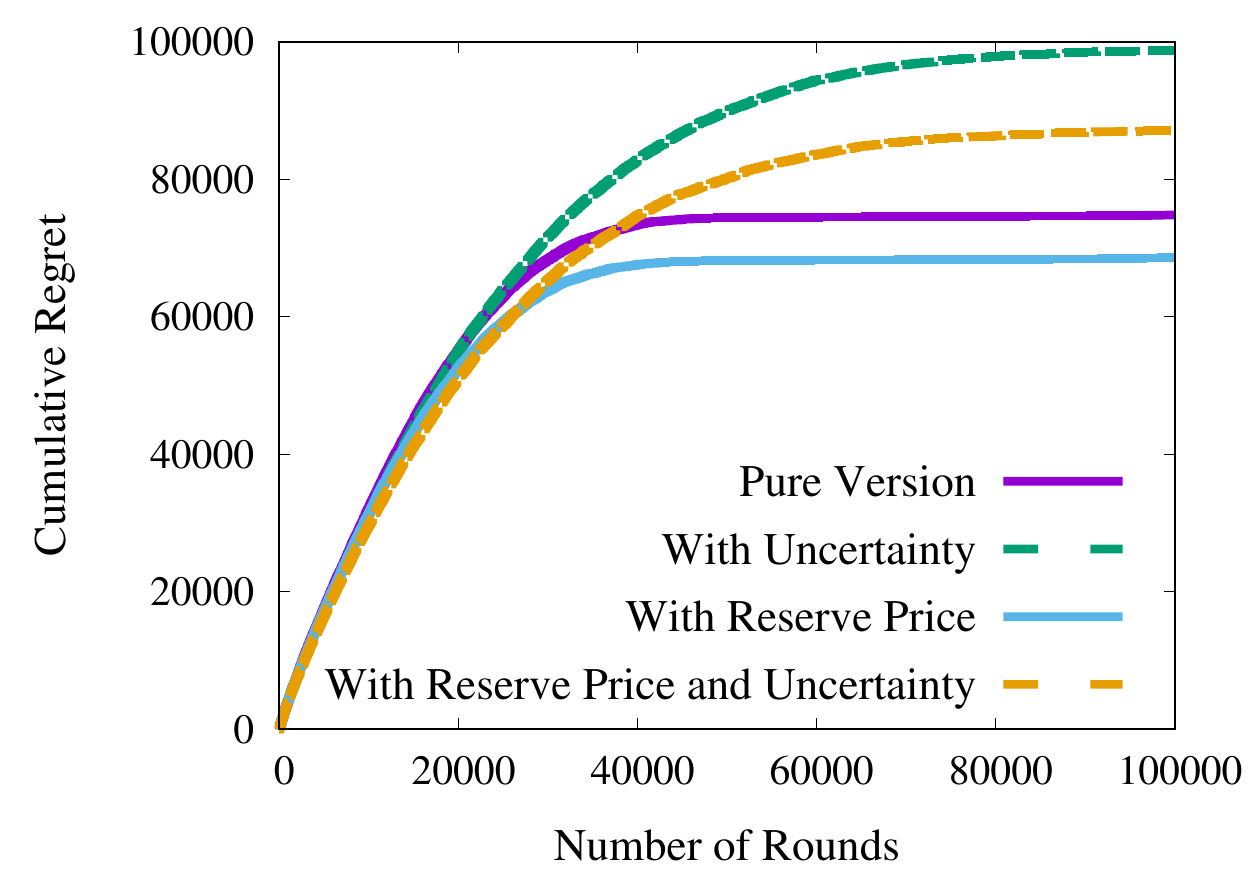}}
\caption{Cumulative regrets with varying dimensions of feature vector in pricing of noisy linear query.}\label{fig:query:total:all}
\vspace{-0.4em}
\end{figure*}

In this section, we present the evaluation results of our pricing mechanism from practical regret and overhead.


We use three real-world datasets, including MovieLens 20M dataset~\cite{link:movielens20m}, Airbnb listings in U.S. major cities~\cite{link:airbnb}, and Avazu mobile ad click dataset~\cite{link:avazuctr}, to evaluate our pricing mechanism over noisy linear queries, accommodation rentals, and impressions, under the linear, log-linear, logistic market value models, respectively. First, the MovieLens dataset contains $20,000,263$ ratings of $27,278$ movies made by $138,493$ anonymous users between January 9, 1995 and March 31, 2015. Second, the Airbnb dataset provides a list of $74,111$ booking records from year 2008 to year 2017 in 6 U.S. cities, \ie, New York, Los Angeles, San Francisco, Washington, Chicago, and Boston. Besides, each record contains a user id, the logarithmic lodging price, house type, location, amenities, host response rate, cancellation policy, and so on. Third, the Avazu dataset comprises 10 days of click-through data, in total $404,289,670$ ad displaying samples. Each sample covers information of both the ad and corresponding mobile user, \eg, an ad id, click or non-click reaction, position, device id, device ip, and internet access type.




\subsection{Application 1: Noisy Linear Query}


We first introduce the details of our setup for trading noisy linear queries, where a static market framework using marked prices has been investigated in~\cite{jour:cacm:2017:li}. On one hand, we regard the MovieLens users, who contributed the ratings, as the data owners in data markets. Additionally, for the data owners, we adopt the differential privacy based privacy leakage quantification mechanism and the tanh based privacy compensation functions from~\cite{jour:cacm:2017:li}. On the other hand, we simulate the noisy linear queries from online data consumers. To validate the adaptivity of our pricing mechanism, the parameters of each linear query are randomly drawn from either a multivariate normal distribution with zero mean vector and identity covariance matrix or a uniform distribution within the interval $[-1, 1]$, while the variance of Laplace noise added to the true answer is randomly selected from $\{10^{k}|k \in \mathbb{Z}, |k|\leq 4\}$. Under each query $Q_t$, we compute the privacy compensations of all data owners, and then generate an $n$-dimensional feature vector with the method in Section~\ref{subsec:problem:formulation}. We recall that we evenly divide the sorted privacy compensations into $n$ partitions, sum those falling into a certain partition, and thus obtain a feature. For the sake of normalization, we scale each feature vector such that its $L_2$ norm is 1, \ie, $\forall t \in [T], \|\mathbf{x}_t\| = 1$ and $S = 1$. Besides, we set the reserve price of a query to be the total privacy compensation, \ie, $q_t = \sum_{i \in [n]} x_{t, i}$ here. In nature, the $L_2$ norm of the weight vector for deriving $q_t$ is $\sqrt{n}$. Moreover, we draw the weight vector $\theta^*$ for modeling the market values of queries in a similar way to sample the query parameters. The difference is that we further scale $\theta^*$ such that its $L_2$ norm is $\sqrt{2n}$, \ie, $\|\theta^*\| = \sqrt{2n}$. This guarantees that the market value of each query $v_t = {\mathbf{x}_t}^T\theta^*$ is no less than its reserve price $q_t$ with a high probability. Furthermore, we set the data broker's initial knowledge set $\mathcal{E}_1$ of $\theta^*$ such that its $L_2$ norm is no more than $2\sqrt{n}$, \ie, $R = 2\sqrt{n}$.



\begin{table}[!t]
\caption{Statistics over pricing of noisy linear query per pound under the version with reserve price.}
\label{tab:query:pricing:statistics}
\centering
\resizebox{\columnwidth}{!}{
\begin{tabular}[t]{rc|cccc}
\toprule
$n$ & $T$ & Market Value & Reserve Price &  Posted Price & Regret\tabularnewline
\midrule\midrule
1 & $10^2$ & 1.414 & 1 & 1.409 (0.045) & 0.035 (0.202)\tabularnewline
20  & $10^4$ & $^*$3.874 (1.278) & 3.388 (0.776) & 3.685 (1.631) & 0.166 (0.824) \tabularnewline
40  & $10^4$ & 5.246 (1.616)     & 4.739	(1.188) & 5.254 (1.614) & 0.743	(1.933) \tabularnewline
60  & $10^5$ & 7.098 (1.910)     & 5.733 (1.491) & 7.089 (1.912) & 0.220 (1.257) \tabularnewline
80  & $10^5$ & 7.266 (2.046)     & 6.531 (1.761) & 7.243 (2.091) & 0.387 (1.690) \tabularnewline
100	& $10^5$ &  8.824 (2.235) 	& 7.221 (1.985) & 8.820 (2.242) & 0.686 (2.461) \tabularnewline
\toprule
\multicolumn{6}{l}{$^*$The entry is stored in the format: mean (standard deviation).} \tabularnewline
\end{tabular}
}
\end{table}


In Fig.~\ref{fig:query:total:all}, we plot the cumulative regrets of four different versions of our pricing mechanism under the linear model, including Algorithm~\ref{alg:ellipsoid:pricing}* (the pure version), Algorithm~\ref{alg:ellipsoid:pricing:noise}* (the version with uncertainty), Algorithm~\ref{alg:ellipsoid:pricing} (the version with reserve price), and Algorithm~\ref{alg:ellipsoid:pricing:noise} (the version with reserve price and uncertainty). Here, the dimension of feature vector $n$ first takes the value 1, and then increases from 20 to 100 with a step of 20. Besides, $\delta$ is fixed at 0.01, which is in the pre-analyzed order of $O(n/T)$ for $n = 1$, but is much larger than $O(n/T)$ for $n \neq 1$. Moreover, in each round $t$, the randomness $\delta_t$ in the market value $v_t$ is drawn from the normal distribution with mean 0 and standard deviation $\sigma = \delta/(\sqrt{2\log 2}\log T)$. Furthermore, the threshold $\epsilon$ is set to $\log_2(T)/T$ for $n = 1$ as claimed in Theorem~\ref{them:one:dimension}, while is set to $n^2/T$ for $n \neq 1$ in Theorem~\ref{theorem:alg2}. As a complement to Fig.~\ref{fig:query:total:all}, we list some precise statistic information about the version with reserve price in Table~\ref{tab:query:pricing:statistics}, including the dimension of feature vector $n$, the number of total rounds $T$, together with the means and standard deviations of market value, reserve price, posted price, and regret per round. We note that the market value column can work as a baseline for relatively measuring the levels of uncertainty (in the magnitude of $0.1\%$ of the market value) and regret.


We first observe Fig.~\ref{fig:query:total:all} holistically. We can see that under a specific version, the cumulative regret at the end of a certain number of rounds $t$ increases with $n$. The reason is that when $n$ grows, the data broker needs to post exploratory prices in more rounds to obtain a good knowledge of the weight vector $\theta^*$, and thus can accumulate more regret. This outcome conforms to our theoretic regret analysis in Section~\ref{sec:regret:analysis}.

\begin{figure*}[t]
\centering
\subfigure[Noisy Linear Query ($n = 100$)]{\label{fig:query:regretn100:ratio}
\includegraphics[width=0.66\columnwidth]{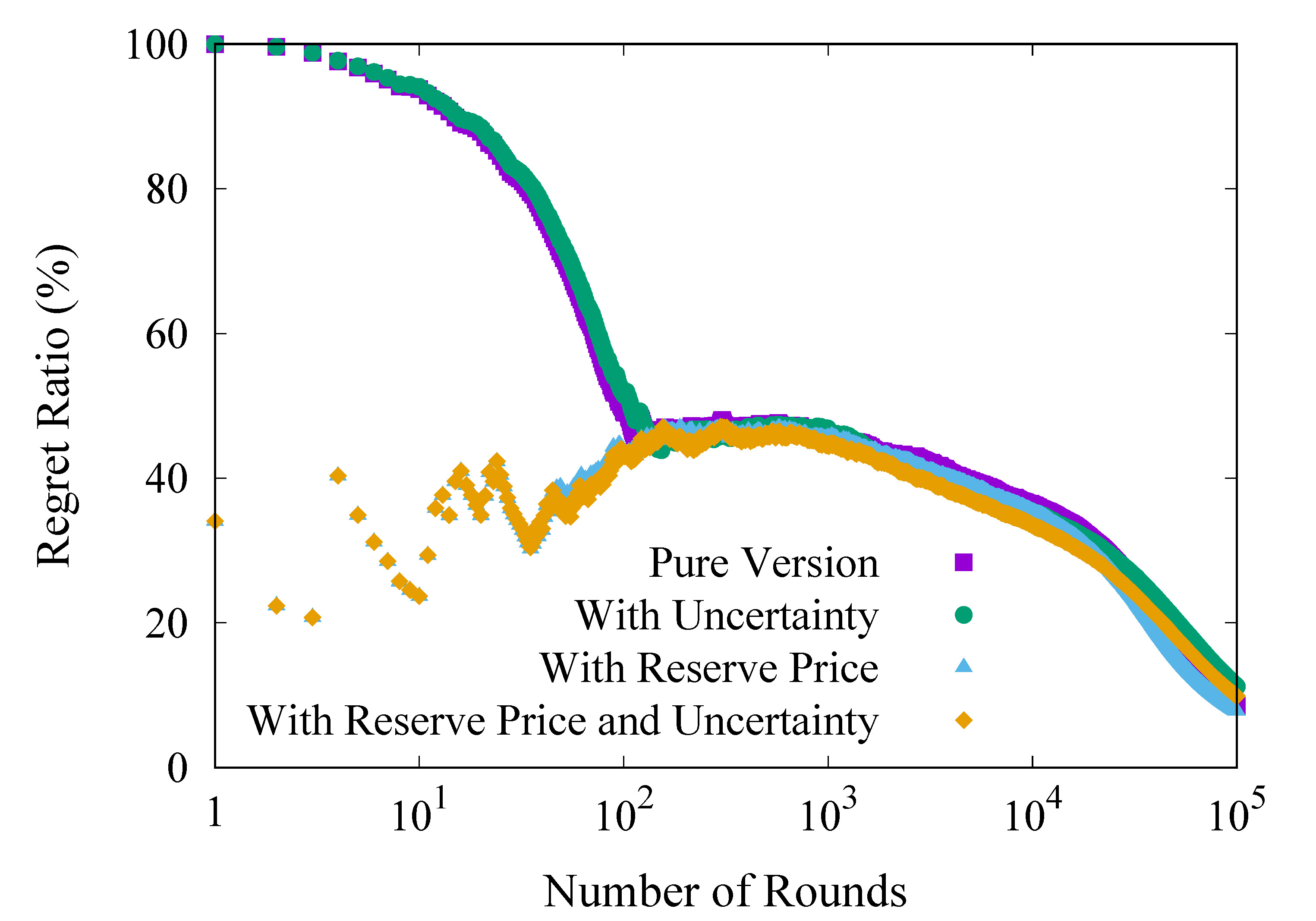}}
\subfigure[Accommodation Rental]{\label{fig:airbnb:regret:ratio}
\includegraphics[width=0.66\columnwidth]{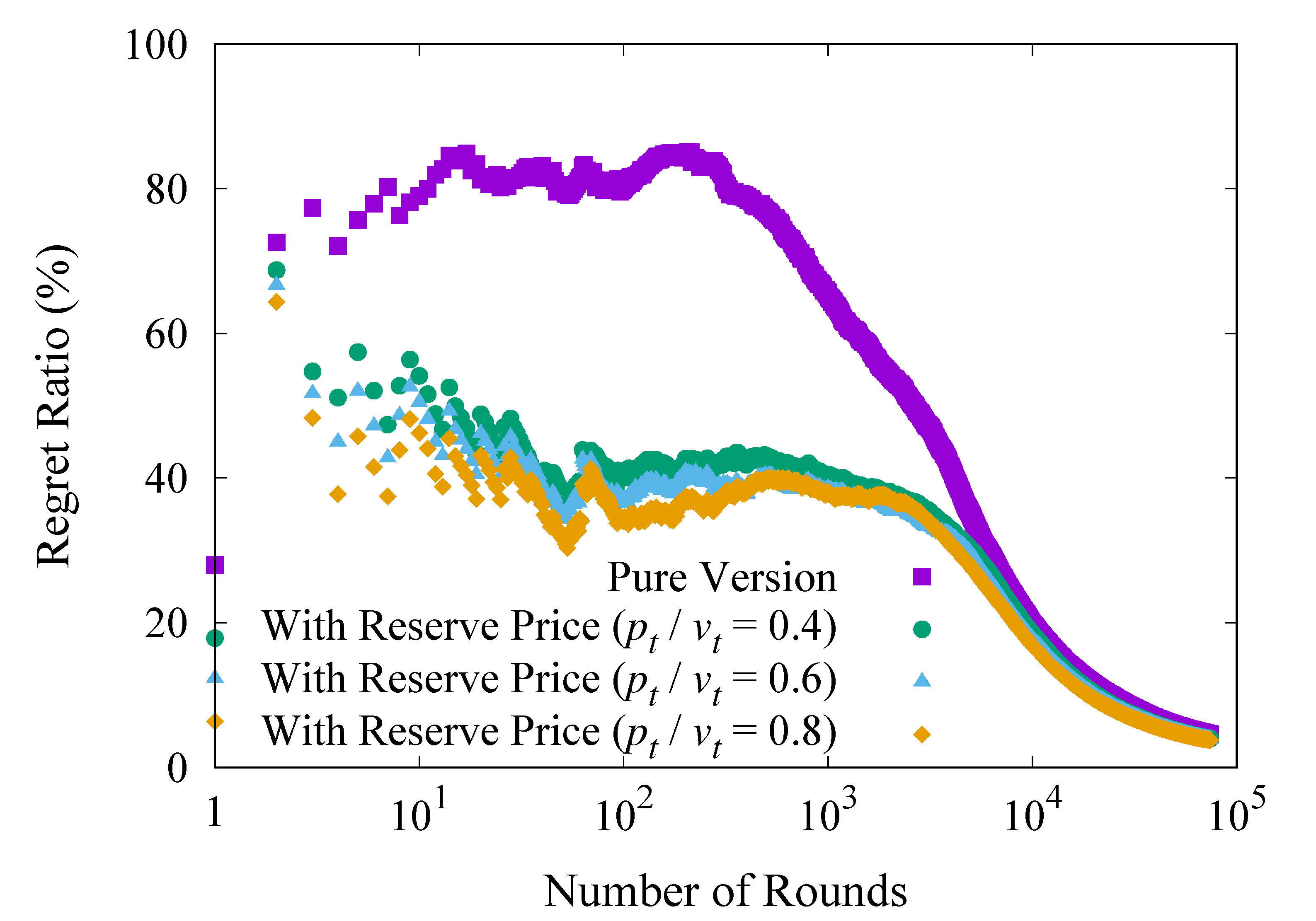}}
\subfigure[Impression]{\label{fig:ctr:regret:ratio}
\includegraphics[width=0.66\columnwidth]{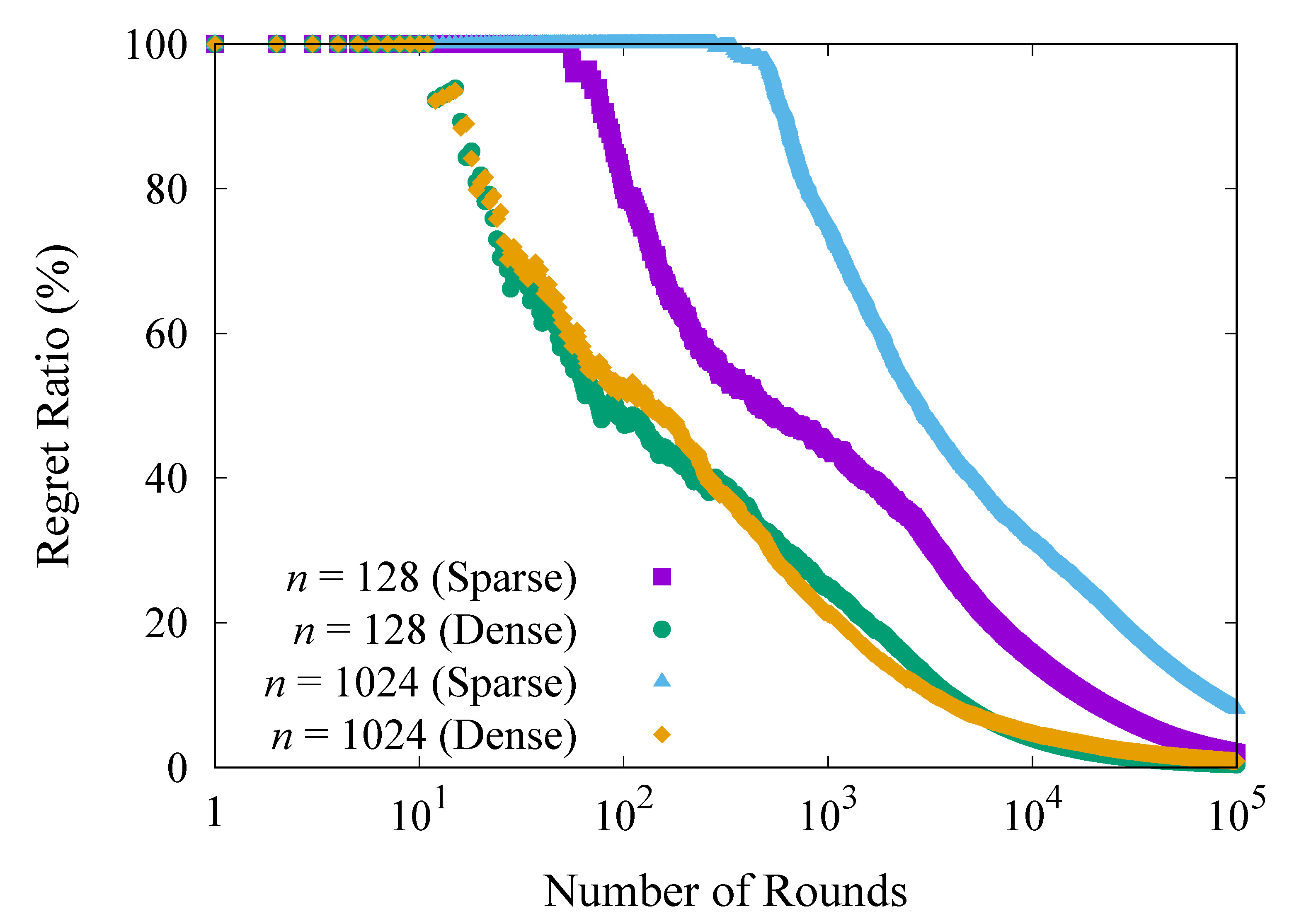}}
\caption{Regret ratios in pricings of noisy linear query, accommodation rental, and impression.}
\end{figure*}

We now observe the one-dimensional case in Fig.~\ref{fig:query:n1} and the multi-dimensional cases from Fig.~\ref{fig:query:n20} to Fig.~\ref{fig:query:n100} more carefully. We start with the one-dimensional case. First, by comparing the purple and blue solid lines in Fig.~\ref{fig:query:n1}, we can see that the introduction of the reserve price constraint has no effect on the pure version of our pricing mechanism. When $n = 1$, the reserve price and the market value of each query are constants 1 and $\sqrt{2}$, respectively. In addition, the data broker's initial knowledge of the market value is the interval $[0, 2]$. Thus, in the first round, no matter the data broker considers or ignores the reserve price 1, she posts the exploratory price 1, which is less than the market value $\sqrt{2}$, and should be accepted by the data consumer. After this round, the interval is refined to $[1, 2]$, which indicates that the reserve price 1 can no longer affect the posted prices. Second, by comparing the purple (\resp, blue) solid line with the green (\resp, golden) dashed line in Fig.~\ref{fig:query:n1}, we can also see that the introduction of low uncertainty will slightly increase the cumulative regret in the pure version (\resp, the version with reserve price).


We next focus on the multi-dimensional cases. Once again, we examine how the reserve price constraint can affect our posted price mechanism. We can find that the incorporation of reserve price can dramatically reduce the cumulative regret. In particular, when $n = 20$ and the number of rounds $t$ is $10^4$, the version with reserve price (\resp, the version with reserve price and uncertainty) reduces $13.16\%$ (\resp, $10.92\%$) of the cumulative regret than the pure version (\resp, the version with uncertainty). We further examine the impact of uncertainty. We can see that the existence of uncertainty accumulates more regret, especially when $t$ is large. In particular, when $n = 60$ and $t = 10^5$, the version with uncertainty (\resp, the version with reserve price and uncertainty) increases $25.25\%$ (\resp, $24.86\%$) of the cumulative regret than the pure version (\resp, the version with reserve price). This is because in the case of a large $t$, the data broker already has a good knowledge of the weight vector $\theta^*$, and posts the conservative price with a high probability. Besides, we recall that to circumvent uncertainty, the conservative price, involving the minimum possible market value $\ubar{p}_t$, decreases by a ``buffer" of size $\delta$ to keep its acceptance ratio (Algorithm~\ref{alg:ellipsoid:pricing:noise}, Line 27), which can generate a higher regret.


We finally provide an intuitive sense of the regret level of our pricing mechanism. We introduce a metric, called {\em regret ratio}, which is defined as the ratio between the cumulative regret and the cumulative market value, \ie, $\sum_{k=1}^{t} R_k / \sum_{k=1}^{t} v_k$ at the end of $t$ rounds. For example, in Table~\ref{tab:query:pricing:statistics}, we can divide the mean values in the regret column by those in the market value column, and get the regret ratios of the version with reserve price for different $n$'s at the end of $T$ rounds .

Coupled with Fig.~\ref{fig:query:n100}, which depicts the cumulative regrets of four different versions for $n = 100$ at the end of different numbers of rounds, Fig.~\ref{fig:query:regretn100:ratio} further plots the regret ratios. One key observation from Fig.~\ref{fig:query:regretn100:ratio} is that when the number of rounds $t$ is small, \ie, when $t \leq 100$, the regret ratio of the version with reserve price (\resp, the version with reserve price and uncertainty) is much lower than that of the pure version (\resp, the version with uncertainty). This reflects a critical functionality of reserve price: it can mitigate the cold-start problem in a posted price mechanism. More specifically, in the beginning, the data broker holds a broad knowledge set of the weight vector $\theta^*$, and thus the market value estimation of a query is coarse, which implies a high regret ratio. However, with the help of reserve price, the data broker can improve the market value estimation, through imposing an additional lower bound and refining the knowledge set more quickly. The mitigation of cold start can be a factor underlying our above observation that the reserve price constraint reduces the cumulative regret.


The second key observation from Fig.~\ref{fig:query:regretn100:ratio} is that as the number of rounds $t$ grows, the difference between the regret ratios of the versions with and without reserve price becomes smaller. Besides, when $t$ is very large, the regret ratios of all four versions are very low. In particular, at the end of $T = 10^5$ rounds, the regret ratios of the pure version, the version with uncertainty, the version with reserve price, and the version with reserve price and uncertainty are $8.48\%$, $11.19\%$, $7.77\%$, and $9.87\%$, respectively. The reason is that after enough rounds, the data broker will have a good estimation of a query's market value, and the effect of reserve price on the posted price diminishes. An extreme example happens in the one-dimensional case presented above, where after the first round, the reserve price has already been excluded from the estimated interval. At last, we provide a risk-averse baseline for the versions involving the reserve price, which consistently posts the reserve price in each round. The regret ratio of such a baseline is $18.16\%$. Therefore, compared with this baseline, the data broker, using our pricing mechanism, can further reduce $57.19\%$ and $45.64\%$ of the regret ratios in the version with reserve price and the version with reserve price and uncertainty, respectively.

These evaluation results demonstrate that our pricing mechanism under the fundamental linear model can indeed reduce the practical regret of the data broker in online data markets.


\subsection{Application 2: Accommodation Rental}

We first describe how to preprocess the Airbnb dataset, and also present some setup details for pricing accommodation rentals under the log-linear model. First, to obtain the feature vector of each booking record, we mainly utilize a data analysis library in Python, called pandas. In particular, we process the categorical features with the pandas built-in data type ``categoricals", which can handle the missing values, and return an integer array of codes for all categories. Besides, we add some interaction features to enhance model capacity. The final dimension of each feature vector $n$ is 55. Second, to obtain the weight vector $\theta^*$ in modeling the market values of accommodations, we regard the logarithmic lodging prices as target variables in supervised learning, and then apply linear regression to learn the coefficients of different features, which play the role of $\theta^*$ here. Specifically, the mean squared error (MSE) over the test set, which occupying $20\%$ of the Airbnb dataset, is $0.226$. Third, to investigate how different settings of reserve price can affect the posted price mechanism, we vary the ratio between the natural logarithms of reserve price and market value, namely $q_t/v_t$. Fourth, when computing the regret ratios, we use the real rather than the logarithmic posted price and market value, by applying the natural exponential function to $p_t$ and $v_t$.


Fig.~\ref{fig:airbnb:regret:ratio} depicts the regret ratios of the pure version of our pricing mechanism under the log-linear model, as well as the versions with reserve price, where $q_t/v_t$ ranges from 0.4, to 0.6, and to 0.8. From Fig.~\ref{fig:airbnb:regret:ratio}, we can see that when the reserve price is set to be closer to the market value, the regret ratio decreases, especially when the number of rounds $t$ is small, \ie, when $t \leq 10^3$. In other words, as the reserve price is approaching the market value, its impact on mitigating the cold-start problem in a posted price mechanism is more evident. We can also see from Fig.~\ref{fig:airbnb:regret:ratio} that at the end of $T = 74,111$ rounds, the regret ratios are very low. In particular, the regret ratios of the pure version, and the versions with reserve price where $q_t/v_t = 0.4, 0.6,$ and $0.8,$ are $4.57\%$, $4.01\%$, $3.83\%$, and $3.79\%$, respectively. We here still consider the risk-averse baseline, where the reserve price is posted in each round, for comparison. The regret ratios of this baseline are $23.40\%$, $17.00\%$, and $9.33\%$ in the versions with reserve price where $q_t/v_t$ = 0.4, 0.6, and 0.8, respectively. Compared with this baseline, the data broker can further reduce $82.88\%$, $77.46\%$, and $59.39\%$ of the regret ratios, respectively.

The above fine-grained evaluation results provide a deeper understanding of the reserve price's role in reducing the practical regret of a posted price mechanism. Besides, our proposed pricing mechanism significantly outperforms the baseline which merely exploits the reserve price.


\subsection{Application 3: Impression in Advertising}

We first introduce data preprocessing and setup for pricing impressions under the logistic model. First, to handle the categorical data fields in ad displaying samples, we utilize one-hot encoding with the hashing trick, where the dimension of the feature vector $n$ serves as the modulus after hashing. Second, we regard the click/non-click states as target variables, further apply Follow The Proximally Regularized Leader (FTRL-Proximal) based logistic regression algorithm (which has been deployed at Google's advertising platform~\cite{proc:kdd13:mcmahan}), and thus obtain the weight vector $\theta^*$ for capturing CTRs. In particular, FTRL-Proximal is an online learning algorithm with per-coordinate learning rates and $L_1, L_2$ regularizations, and can preserve excellent performance and sparsity. When testing over the samples in the last two days, the logistic loss is 0.420 (\resp, 0.406) at $n = 128$ (\resp, $n = 1024$). Besides, the learnt weight vector $\theta^*$ is quite sparse. Specifically, the number of nonzero elements in $\theta^*$ is 21 (\resp, 23) at $n = 128$ (\resp, $n = 1024$). In what follows, we investigate two different cases to validate the feasibility of our pricing mechanism over both sparse and dense feature vectors. One is the sparse case, which keeps all the elements in each feature vector no matter their weights are nonzero or zero. The other is the dense case, which omits those features if their weights are zero.

In Fig.~\ref{fig:ctr:regret:ratio}, we plot the regret ratios of the pure version of our pricing mechanism in both sparse and dense cases, when $n$ takes the values $128$ and $1024$. We can observe from Fig.~\ref{fig:ctr:regret:ratio} that the regret ratio in the sparse case decreases more slowly than that in the dense case, especially when the number of rounds $t$ is smaller than $10^3$. This outcome stems from that the starting rounds are mainly dedicated to eliminating those zero elements in the weight vector, which implies a larger regret ratio in the beginning. This reason can also account for the phenomenon that in the sparse case, the regret ratio for $n = 1024$ decreases more slowly than that for $n = 128$. Even so, after $10^5$ rounds, the regret ratios are $2.02\%$, $0.41\%$, $8.04\%$, and $0.89\%$, when $n$ takes the values $128$ and $1024$ in the spare and dense cases, respectively. Furthermore, when the number of rounds $t$ becomes larger, the regret ratios can be sustainably reduced.

These evaluation results reveal that our proposed pricing mechanism performs well over both sparse and dense feature vectors. By further combining with the pricing of accommodation rental, we can conclude that our pricing mechanism has a good extensibility to non-linear market value models.

\subsection{Details on Implementation and Overhead}\label{subsec:implementation:overhead}

We implemented our pricing mechanism in Python $2.7.15$. The running environment is a Broadwell-E workstation with 64-bit Ubuntu $16.04.5$ Linux operation system. In particular, the processor is Intel(R) Core(TM) i7-6900K with 8 cores, the base frequency is 3.20GHz, the memory size is 64GB, and the cache size is 20MB. We use the function {\sf time.time()} in Python to record the elapsed time, and use the command {\sf cat /proc/PID/status $|$ grep 'VmRSS'} to monitor memory overhead, where {\sf PID} is the process identifier of our program. Our source code is online available from~\cite{link:source:code}.

We next present the computation and memory overheads of the above three applications. For the pricing of noisy linear query under the version with reserve price, when $n = 100$, the online latency of the data broker in determining the posted price and updating her knowledge set is 0.115ms per query. Besides, the memory overhead is 151MB. For the pricing of accommodation rental under the version with reserve price ($q_t/v_t = 0.6$), the online latency is 0.019ms per booking request, and the memory overhead is 105MB. For the pricing of impression, when $n = 1024$, the online latency is 3.509ms (\resp, 0.024ms) per ad displaying sample in the sparse (\resp, dense) case. Additionally, the memory overhead is 106MB (\resp, 75MB) in the sparse (\resp, dense) case.

In a nutshell, our proposed pricing mechanism has a light load under both linear and non-linear models, and can be employed to dynamically price those products with customization, existence of reserve price, and timeliness properties.


\section{Related Work}\label{sec:related:work}
In this section, we briefly review related work.

\subsection{Data Market Design}



An explosive demand for sharing data contributes to growing attention on data market design~\cite{jour:vldb11:datamarket,proc:balazinska:2013}. The researchers from the database community~\cite{proc:pods:2012:koutris,jour:vldb:2012:arbitrage,proc:sigmod:2013:arbitrage,jour:vldb:2014:lin,proc:sigmod:17:arbitrage,proc:icdt:2017:deep:arbitrage,jour:vldb:2017:arbitrage} mainly focused on arbitrage freeness in pricing queries over the relational databases, which maintain insensitive data. Here, the existence of arbitrage means that the data consumer can buy a query with a lower price than the marked price, through combining a bundle of other cheaper queries. Thus, the data broker needs to rule out arbitrage opportunities to preserve her revenue. Specific to personal data trading, the researchers routinely adopted the cost-plus pricing strategy, where the data broker first compensates each data owner for her privacy leakage, and then scales up the total privacy compensation to determine the price of query for the data consumer. In particular, the quantification of privacy leakage primarily resorts to the classical differential privacy framework~\cite{proc:tcc06:dp,jour:2014:dwork:dp} or its variants, \eg, Pufferfish privacy~\cite{proc:sigmod:2011:kifer:privacy,proc:pods:12:pufferfish,jour:tods:14:pufferfish}. Besides, different researchers investigated distinct types of queries from the data consumers. Ghosh and Roth~\cite{proc:ec11:ghosh:selling:privacy} considered single counting query. The follow-up work by Li~\et~\cite{jour:cacm:2017:li} further extended to multiple noisy linear queries. Hynes~\et~\cite{jour:vldb18:hynes} investigated model training requests. We took the ubiquitous data correlations into account, including the correlations among data owners~\cite{proc:kdd:2018:erato} and the temporal correlations within a certain data owner's time-series data~\cite{proc:infocom:2019:horae}. We thus proposed to trade noisy aggregate statistics over private correlated data.



The original intention of these works is to ensure the consistency and robustness of a non-interactive pricing strategy against cunning data consumers. However, they didn't consider the responses from the data consumers, and further optimize the cumulative revenue of the data broker, which is instead the major focus of this work.

\subsection{Contextual Dynamic Pricing}

The dynamic pricing problem has been extensively studied in diverse contexts. The pioneering work by Kleinberg and Leighton~\cite{proc:focs03:kleinberg} considered markets for identical products, and designed optimal posted pricing strategies under the buyers' identical, random, and worst-case valuation models. However, the products in practical markets, \eg, online commerce and advertising, tend to differ from each other. This further motivated the emergence of contextual pricing, where the seller intends to sell a sequence of highly differentiated products, posts a price for each product, and then observes whether the buyer accepts or not. More specifically, each product is represented by a feature vector for differentiation, while its market value is typically assumed to linear in the feature vector. The researchers thus turned to online learning the weight vector from feedbacks, and further converted this task to a multi-dimensional binary search problem. Amin~\et~\cite{proc:nips14:amin} first proposed a stochastic gradient descent (SGD) based solution, which can attain $O(T^{2/3})$ strategic regret by ignoring logarithmic terms. However, their solution requires the feature vectors to be drawn from an i.i.d. distribution, such that each feature vector can serve as an unbiased estimate of the weight vector. The follow-up work by Cohen~\et~\cite{proc:ec16:cohen} abandoned this strict requirement. Besides, they approximated the polytope-shaped knowledge set with ellipsoid, and provided a worst-case regret of $O(n^2\log T)$, which is essentially the pure version of our pricing mechanism. Lobel~\et~\cite{proc:ec17:leme} further reduced regret to $O(n\log T)$ by projecting and cylindrifying the polytope. Most recently, Leme~\et~\cite{proc:focs18:leme} borrowed a key concept from geometric probability, called the intrinsic volumes of a convex body, and achieved a regret guarantee of $O(n^4\log\log(nT))$. The key principle behind this line of works is to identify the centroid of the knowledge set or its projection/transformation, such that each exploratory posted price can roughly impose a central cut in terms of different measures, \eg, volume, surface area, and width. Besides, although the most recent two works optimized the regret, they are too computationally complex to be deployed in practical online markets.


Yet, it is worth noting that the contextual dynamic pricing mechanisms significantly differ from the classical cutting-plane or localization algorithms in the field of convex optimization~\cite{book:boyd2004convex}, \eg, the original ellipsoid method~\cite{jour:khachiyan:1979} and the analytic center cutting-plane method\cite{jour:mp93:accpm}. In particular, the purpose of a cutting-plane method is to find a point in a convex set for optimizing a preset objective function. In contrast, the goal of a contextual dynamic pricing mechanism is to minimize the cumulative regret during the process of locating a preset point, \ie, the weight vector here. Furthermore, under contextual dynamic pricing, the direction of each cut is fixed by the feature vector of a product requested by a buyer, while the seller can only choose the position of the cut through posting a certain price. This setting distinguishes contextual dynamic pricing from a majority of ellipsoid based designs~\cite{jour:ms:2003:toubia,jour:2004:olivier,proc:stoc16:roth}, which allow the seller to control the direction of each cut. In fact, the contextual dynamic pricing problem can also be modeled into a contextual multi-armed bandit (MAB), where the arms/actions to be exploited and explored are the domain of the weight vector. However, given the domain of weight vector is continuous, we need to apply the discretization technique, which makes the number of bandits extremely large. In addition to inefficiency, since the payoff/regret function is piecewise and highly asymmetric, this sort of solutions can be oracle-based, \eg, \cite{jour:siam:2002:auer,proc:icml14:agarwal,proc:icml16:syrgkanis,proc:nips16:syrgkanis,proc:focs17:dudik,proc:nips18:foster}, and inevitably incurs polynomial rather than logarithmic regret in the number of total rounds $T$~\cite{proc:ec17:leme}.


Unfortunately, none of above works has incorporated the reserve price constraint, and further examined its impact on a posted price mechanism. In particular, due to the existence of reserve price, the cut over the knowledge set should support an arbitrary position, which instead was not touched in existing works. Besides, the impacts of reserve price on mitigating the cold-start problem and thus reducing the cumulative regret are first analyzed and verified in this work.

\section{Conclusion}\label{sec:conclusion}
In this paper, we have proposed the first contextual dynamic pricing mechanism with the reserve price constraint, for the data broker to maximize her revenue in online personal data markets. Our posted price mechanism features the properties of ellipsoid to perform online optimization effectively and efficiently, and can support both linear and nonlinear market value models, while allowing some uncertainty. We further have extended it to several other similar application scenarios, and extensively evaluated over three practical datasets. Empirical results have demonstrated the feasibility and generality of our pricing mechanism, as well as the functionality of the reserve price constraint.

\section*{Acknowledgment}
This work was supported in part by Science and Technology Innovation 2030 - ``New Generation Artificial Intelligence" Major Project No. 2018AAA0100905, in part by China NSF grant 61972252, 61972254, 61672348, and 61672353, in part by the Open Project Program of the State Key Laboratory of Mathematical Engineering and Advanced Computing 2018A09, and in part by Alibaba Group through Alibaba Innovation Research (AIR) Program. The opinions, findings, conclusions, and recommendations expressed in this paper are those of the authors and do not necessarily reflect the views of the funding agencies or the government. Fan Wu is the corresponding author.


\bibliographystyle{IEEEtran}
\bibliography{reference}  


\section*{appendix}
\begin{theorem}\label{them:one:dimension}
The worst-case regret of the pure version of our pricing mechanism in the one-dimensional case is $O(\log T)$.
\end{theorem}
\begin{proof}
As illustrated in Section~\ref{sec:design:principles}, ellipsoid degenerates to interval in the one-dimensional case. Specifically, we let $\mathbf{x}_t \in [-S, S]$ and $\mathcal{K}_1 = \{\theta \in \mathbb{R}|\theta \in [-R, R]\}$. Thus, we have: $\forall t \in [T], \bar{p}_t - \ubar{p}_t \leq 2RS$. Besides, in round $t$, if $\bar{p}_t - \ubar{p}_t > \epsilon$, the data broker will take the exploratory price to refine her knowledge set of $\theta^*$ by half. Hence, we have: $T_e \leq \log_2(2RS/\epsilon)$. Moreover, in round $t$, the regret incurred by posting the exploratory (\resp, conservative) price is bounded above by $\bar{p}_t$ (\resp, $\bar{p}_t - \ubar{p}_t$), and can be further bounded above by $RS$ (\resp, the threshold $\epsilon$). Therefore, the total regret is no more than $T_e RS + (T - T_e)\epsilon$. When $T_e$ takes its upper bound $\log_2(2RS/\epsilon)$ and $\epsilon$ takes $\log_2(T)/T$, the cumulative regret is $O(\log T)$. This completes the proof.
\end{proof}

\begin{figure}[t]
\centering
\includegraphics[width = 0.8\columnwidth]{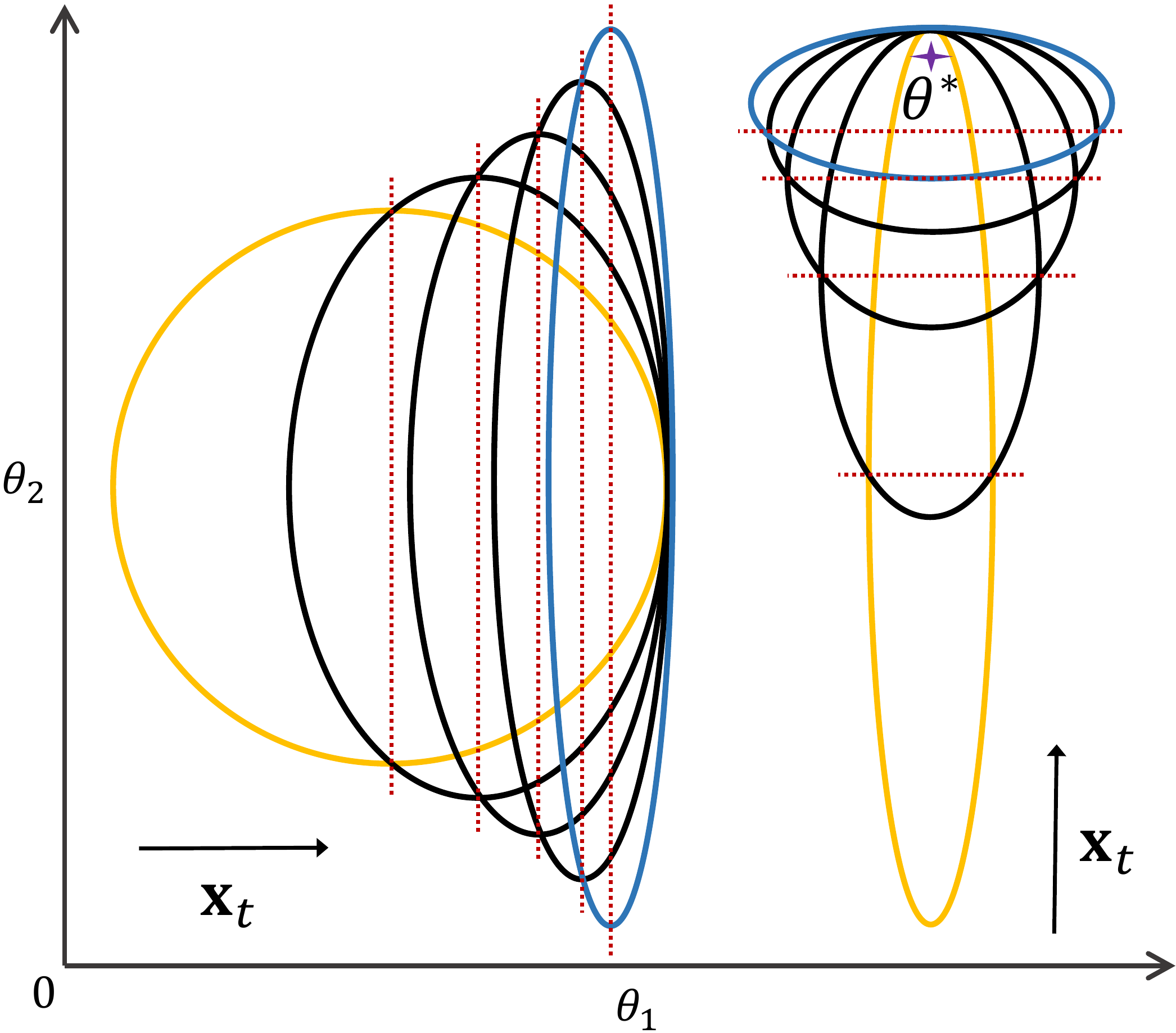}
\caption{An adversarial example in the two-dimensional case, if the conservative posted prices are allowed to cut the ellipsoid.}\label{fig:adversary:case}
\end{figure}

\begin{lemma}\label{lem:adversarial:example}
If the data broker are allowed to utilize the conservative posted prices to refine the knowledge set, her cumulative worst-case regret is $O(T)$.
\end{lemma}
\begin{proof}
We set $R = 1$ and $S = 1$ for clarity, \ie, $\mathbf{A}_1 = \mathbf{I}_{n \times n}$ and $\forall t \in [T], \|\mathbf{x}_t\| = 1$. We consider an adversary, who chooses a sequence of $T$ queries, where the feature vectors in the first half are $\forall t \in [\lfloor T/2 \rfloor], \mathbf{x}_t = (1, 0, 0, \ldots, 0)^T$, and the feature vectors in the second half are $\forall t \in \{\lfloor T/2\rfloor + 1, \ldots, T\}, \mathbf{x}_t = (0, 1, 0, \ldots, 0)^T$. Besides, in the first half of rounds, the adversary sets the reserve prices to be the middle prices, \ie, $\forall t \in [\lfloor T/2 \rfloor], q_t = \frac{\ubar{p}_t + \bar{p}_t}{2}$, while discarding the reserve price constraint in the second half. Under the above adversary setting, suppose the data broker is allowed to utilize the conservative prices to refine her knowledge set. During the first $\lfloor T/2\rfloor$ rounds, she will successively impose central cuts over the ellipsoid along the first coordinate. Besides, the width of the ellipsoid's first axis shrinks exponentially, while all the lengthes of the other $n - 1$ axes expand exponentially, in particular with the ratio $\frac{n}{\sqrt{n^2 - 1}}$. Thus, at the end of $\lfloor T/2 \rfloor$ rounds, the width of the ellipsoid along the second coordinate is $2(\frac{n}{\sqrt{n^2 - 1}})^{\lfloor T/2 \rfloor}$. We now compute the cumulative regret in the second half. Given $\bar{p}_t - \ubar{p}_t$ represents the width of the ellipsoid along $\mathbf{x}_t = (0, 1, 0, \ldots, 0)^T$, the data broker will successively post the exploratory prices, \ie, the middle prices, to cut along the second coordinate, until the width is below the threshold $\epsilon$. Besides, each cut can at most shorten the width by half. Hence, the number of rounds where the exploratory prices are posted is at least $\min(T - \lfloor T/2\rfloor, \log_2(2(\frac{n}{\sqrt{n^2 - 1}})^{\lfloor T/2 \rfloor} / \epsilon)) = O(T)$. Therefore, the cumulative worst-case regret in the above adversary case is $O(T)$, which completes the proof. In Fig.~\ref{fig:adversary:case}, we give an intuitive view of this adversarial example in the two-dimensional case.
\end{proof}

\end{document}